\documentclass[aps,pra,twocolumn,longbibliography,superscriptaddress]{revtex4-1}

\usepackage{silence}
\usepackage{adjustbox}
\usepackage[utf8]{inputenc} 
\usepackage[T1]{fontenc}    
\usepackage{hyperref}       
\usepackage{url}            
\usepackage{booktabs}       
\usepackage{amsfonts}       
\usepackage{nicefrac}       
\usepackage{microtype}      
\usepackage{natbib}
\usepackage[normalem]{ulem}
\usepackage{bm}
\usepackage{braket}
\usepackage{graphicx}
\usepackage{amsthm}
\usepackage{algorithmic}
\usepackage[linesnumbered,ruled]{algorithm2e}
\newtheorem{thm}{Theorem}

\newtheorem{prop}[thm]{Proposition}

\usepackage{mathtools}
\usepackage{nccmath}
\usepackage{xcolor}
\usepackage{caption}
\usepackage{tikz}
\usetikzlibrary{quantikz}
\tikzstyle{tensor}=[rectangle,draw=blue!50,fill=blue!20,thick]
\tikzstyle{input}=[circle,draw=blue!50,fill=blue!20,thick]
\usepackage{tikz-network}

\newcommand{\tl}[1]{\textcolor{violet}{#1}}

\begin{document}
\title{The Expressive Power of Parameterized Quantum Circuits}
	
\author{Yuxuan Du}\email{yudu5543@uni.sydney.edu.au}
\affiliation{UBTECH Sydney Artificial Intelligence Centre and the School of Information Technologies, Faculty of Engineering and Information Technologies, The University of Sydney, Australia}
\author{Min-Hsiu Hsieh} \email{min-hsiu.hsieh@uts.edu.au}
\affiliation{Centre for Quantum Software and Information, Faculty of Engineering and Information Technology, University of Technology Sydney, Australia}
\author{Tongliang Liu}
\email{tongliang.liu@sydney.edu.au}
\affiliation{UBTECH Sydney Artificial Intelligence Centre and the School of Information Technologies, Faculty of Engineering and Information Technologies, The University of Sydney, Australia}
\author{Dacheng Tao}
\email{dacheng.tao@sydney.edu.au}
\affiliation{UBTECH Sydney Artificial Intelligence Centre and the School of Information Technologies, Faculty of Engineering and Information Technologies, The University of Sydney, Australia}
	
\date{\today}
	
\begin{abstract}
Parameterized quantum circuits (PQCs) have been broadly used as a hybrid quantum-classical machine learning scheme to accomplish generative tasks. However, whether PQCs have better expressive power than classical generative neural networks, such as restricted or deep Boltzmann machines, remains an open issue. In this paper, we prove that PQCs with a simple structure already outperform any classical neural network  for generative tasks, unless the polynomial hierarchy collapses. Our proof builds on known results from tensor networks and quantum circuits (in particular, instantaneous quantum polynomial circuits). In addition, PQCs equipped with ancillary qubits for post-selection have even stronger expressive power than those without post-selection. We employ them as an application for Bayesian learning, since it is possible to learn prior probabilities rather than assuming they are known. We expect that it will find many more applications in semi-supervised learning where prior distributions are normally assumed to be unknown.  Lastly, we conduct several numerical experiments using the Rigetti Forest platform to demonstrate the performance of the proposed Bayesian quantum circuit.

\end{abstract}
	
	
	\maketitle
	
	\section{Introduction}
There is a ubiquitous belief called 'quantum supremacy' that quantum computers will outperform classical computers \cite{harrow2017quantum}. One characterization of quantum supremacy relates to the expressive power of quantum computing, since the probability distribution generated by quantum devices may not, classically, be sampled efficiently and accurately. Two leading proposals toward this goal  are Boson sampling \cite{aaronson2011computational} and instantaneous quantum polynomial time (IQP) circuits \cite{shepherd2009temporally}.

The system noise in current implementations is known to be the major roadblock. Widespread explorations have been conducted to verify whether noisy intermediate-scale quantum (NISQ) \cite{preskill2018quantum} devices can also outperform classical computers for specific computation tasks. It has been proved that, with system noise, quantum supremacy will disappear in Boson sampling  \cite{neville2017classical} but will remain in IQP \cite{bremner2017achieving}. In addition to demonstrating the existence of quantum supremacy, the issue of finding practical applications for NISQ devices with quantum advantages needs to be {further studied}.
	
Quantum machine learning problems have been popularized because of their ability to efficiently process tremendous amounts of data. They are also exploited as alternative testbeds to confirm quantum advantages \cite{biamonte2017quantum,rebentrost2014quantum,lloyd2014quantum,wittek2014quantum,wiebe2012quantum,paparo2014quantum,dunjko2016quantum}.  {By employing NISQ devices}, potential quantum advantages {may still be retained}, benefiting from the fact that most statistical {machine} learning {algorithms} are robust to system noise, i.e., the noise contained in the input data and models has a negligible influence on the final results \cite{dietterich2000ensemble,tipping2001sparse}.
	
Expressive power is a central topic in classical machine learning and it has generated great interest in quantum machine learning. It is deeply tied to two major topics in machine learning: discriminative modeling and generative modeling, which aim to learn patterns and the probability distribution of input data \cite{ng2002discriminative}{, respectively}. Expressive power in discriminative learning relates strongly to classification performance, e.g., by employing the kernel method \cite{hofmann2008kernel}, the kernel support vector machine (SVM) can efficiently classify nonlinear data. In generative modeling, the expressive power of two highly successful models, restricted Boltzmann machine (RBM) and deep Boltzmann Machine (DBM) \cite{hinton2006reducing,salakhutdinov2010efficient}, to represent quantum many-body states have been extensively investigated  \cite{gao2017efficient1,chen2018equivalence,glasser2018neural,glasser2018supervised,deng2017quantum}. {Consequently, RBM and DBM have been broadly applied to physics research, e.g., identifying phase transition, solving many-body wave functions, and accelerating Monte Carlo simulations} \cite{carrasquilla2017machine,carleo2017solving,liu2017self,liu2017self,van2017learning,ch2017machine}.
	
Parametrized quantum circuits (PQCs) are a promising NISQ scheme that has demonstrated their potential to be applied to practical applications with quantum advantages.
{By employing variational hybrid quantum/classical algorithms, PQCs have been applied to accomplish both the generative \cite{benedetti2018generative,liu2018differentiable,benedetti2017quantum} and discriminative \cite{farhi2018classification,schuld2018circuit,huggins2018towards,grant2018hierarchical} tasks}. PQC is composed of a set of parameterized single and controlled single qubit gates with noise, and the parameters are iteratively optimized by a classical optimizer. In general, the proposed PQCs can be divided into two types: multiple-layer PQCs (MPQCs) and tensor network PQCs (TPQCs). An MPQC consists of multiple blocks {of} quantum circuits in which the arrangement of quantum gates in each block is identical \cite{benedetti2018generative,liu2018differentiable,LaRose2018Variationally}. Mathematically, {we denote} the input quantum state as $\ket{0}^{\otimes N}$ with $N$ qubits, the total number of blocks as $L$, and the $i$-th block as $U(\bm{\theta}^i)$, where the number of parameters is proportional to the number of qubits $|\bm{\theta}|\propto N$ {and $N$ is logarithmically proportional to the dimension of the generated data}. {T}he generated quantum state of MPQC, $\ket{\Phi}$, is defined as $\ket{\Phi} = \prod_{i=1}^{L}U(\bm{\theta}^i)\ket{0}^{\otimes N}$. The tensor PQCs (TPQCs) treat each block as a local tensor. The arrangement of the blocks follows a specified tensor network, such as matrix product states and tree tensor network \cite{huggins2018towards}. Mathematically, the $i$-th block $U(\bm{\theta}^i)$ is composed of $M_i$ local tensor blocks, with $M_i\propto N/2^i$, denoted as $U(\bm{\theta}^i) = \bigotimes_{j=1}^{M_i}U(\bm{\theta}^i_j)$. The generated state from TPQC is defined as $\ket{\Phi} = \prod_{i=1}^{L}\bigotimes_{j=1}^{M_i}U(\bm{\theta}^i_j)\ket{0}^{\otimes N}$. Refer Section \ref{sec:2} for more details.

Although PQCs have provided strong evidence of quantum advantage, \cite{farhi2016quantum,otterbach2017unsupervised}, two important questions remain unexplored: (1) What is the expressive power of PQCs{?} (2) Is there any quantum advantage of PQCs that can be used to solve practical problems? A comparison of expressive power between PQCs and classical neural networks is desirable, and may benefit both physics and machine learning areas, {since PQCs are capable of solving many kinds of machine learning tasks, and classical machine learning methods have also been extensively applied to physics research}.
	
To analyze their relationships, we {will first} prove that MPQCs can be formulated by the tensor network language. This {will} show that MPQCs, TPQCs, and  classical neural networks have a close connection with tensor networks, such as matrix product states (MPS) and multi-scale entanglement renormalization ansatz (MERA) \cite{schollwock2011density,vidal2007entanglement}.  We will then exploit entanglement entropy as a metric to evaluate the expressive power of tensor network states,  to characterize the expressive power of PQCs and neural networks. We will provide a rigorous proof that, given the number {of} trainable parameters that polynomially {scale} with the number of {qubits}, MPQCs,  TQPCs, DBM and long range RBM exhibit volume law entanglement efficiently, while and short range RBM only exhibit area law entanglement efficiently \cite{eisert2010colloquium}.

Before answering the  question of {whether PQCs have any quantum advantages over classical generative algorithms}, {we remark} that entanglement entropy is not the only metric for quantifying expressive power. Even though MPQCs, DBM{,} and long range RBM can efficiently represent quantum states with volume law, {we devise a toy model to prove that some probability distributions can be efficiently generated by MPQCs, DBM, and long range RBM, but the distributions are difficult to be generated {by} TPQCs. We further prove that instantaneous quantum polytime (IQP) circuits} \cite{bernstein1997quantum}\cite{bernstein1997quantum} are a special subclass of MPQCs.  The probability distribution generated by IQP cannot be sampled efficiently and accurately by any classical neural network \cite{shepherd2009temporally}. This indicates that, from the perspective of complexity theory, MPQCs have a stronger expressive power than classical neural networks {and} have the potential to become a practical application with `quantum supremacy' \cite{boixo2018characterizing}.

Finally, we equip MPQCs with ancillary qubits for post-selection---a model we called ancillary driven MPQCs (AD-MPQCs). 	
We show that the class of AD-MPQCs contains post-IQP circuits as a special case.  
Apart from the stronger expressive power, AD-MPQCs  also provide additional benefits from the machine learning perspective. Specifically, AD-MPQCs with a simple structure, that we call  Bayesian quantum circuit (BQC), is devised for Bayesian learning. {Theoretically, we prove that the expressive power of BQC is equivalent to post-IQP.}  {From the machine learning point of view,  the ancillary qubits of BQC can be used to represent the additional information, such as a prior distribution. } BQC not only can exploit priors to improve the performance of a learning task, but can also enable the estimation of prior distributions from the given data, which is highly desired for semi-supervised learning \cite{basu2002semi}. To the best of our knowledge, BQC is the first PQCs that can learn prior distributions from given data. A toy model is designed to verify its effectiveness. The BQC experiments are implemented in Python, leveraging the pyQuil library to access the numerical simulator known as quantum virtual machine (QVM) \cite{smith2016practical}. 

	\section{Definitions and Preliminaries}\label{sec:2}
	
	\subsection{Boltzmann Machine}
	The Boltzmann machine (BM), inspired by {the} Ising model, plays a significant role in the development of the deep neural network, which aims to learn a distribution over the set of their inputs \cite{ackley1985learning,rumelhart1985learning}.
	{Specifically,} BM can be divided into two parts: $N$ visible units $\bm{v}=\{v_i\}_{i=1}^N$ and $M$ hidden units $\bm{h}=\{h_j\}_{j=1}^M$. Given trainable parameters $w_{ij}$ and $b_i$, the Hamiltonian is defined as $H(\bm{s})=\sum_ib_is_i +\sum_{i<j}w_{ij}s_is_j$, with $\bm{s}=\{\bm{v},\bm{h}\}$.  The joint probability  distribution over {the} visible and hidden units is defined as
	\begin{equation}
	P(\bm{v},\bm{h}) = \frac{1}{\mathcal{Z}}e^{-H(\bm{v},\bm{h})}~,
	\end{equation}
	where $\mathcal{Z} = \sum_{\bm{v}}\sum_{\bm{h}}e^{-H(\bm{v},\bm{h})}$ is called {the} partition function. For generative tasks, {the marginal probability distribution of visible units $P(\bm{v})=\sum_{\bm{h}}P(\bm{v},\bm{h})$ is expected to be maximized by optimizing $w_{ij}$ and $b_i$.}
	
	The restricted Boltzmann machine (RBM) \cite{nair2010rectified} is a special type of BM, which can be learned more efficiently. Mathematically, the Hamiltonian of RBM is defined as $H(\bm{v}, \bm{h})=\sum_ib_iv_i +\sum_jb_jh_j +\sum_{i,j}w_{ij}v_ih_j$, where only the inner connections between visible units and hidden units remain. {An} RBM is sparse (or short range), if the connection between visible and hidden units is sparse. {For short range RBM,} the visible unit $v_i$ only connects with $2k+1$ hidden units with a small constant $k$ or $k\sim O(\log M)$. Similarly, {an} RBM is non-sparse (or long range) if $k$ satisfies $k\sim O(M)$.
	
	A deep Boltzmann machine  (DBM) \cite{salakhutdinov2010efficient}, different from a RBM that {includes} only one layer of hidden units, contains many layers of hidden units. In DBM, multiple hidden layers can be learned by training one hidden layer once at a time as for RBM. 
	When {we calculate} the probability distribution between {the} $n$-th layer and $n+1$-th layer, the hidden units of the previous $n$-th layer $\bm{h}^n$ are treated as visible units $\bm{v^n}$ {and $P(\bm{v}^n, \bm{h}^{n+1})$ is obtained as RBM does}.
	
	\subsection{Tensor Networks}\label{sec:2B}
	
	Matrix Product State (MPS) is a natural choice {to efficiently represent} 1D low energy quantum states \cite{schollwock2011density}. {We denote a quantum state of one dimensional lattice with $N$ sites as $\ket{\Psi} = \sum_{j_1,j_2,...,j_N=1}^{d}C_{j_1j_2...j_N}\ket{j_1}\otimes \ket{j_2}\otimes ...\otimes \ket{j_N}$, where all sites have the same dimension $d$.} The state $\ket{\Psi}$ can be completely described by a rank-$N$ tensor $C_{j_1j_2...j_N}$ with total $d^{N}$ elements. {However,  such an exponentially scaling relation implies that} the computation cost {becomes expensive for large $N$}. MPS {enables  $\ket{\Psi}$ to be approximated}  with a high accuracy using only $O(poly(N))$ parameters. {We} rewrite $\ket{\Psi}$ as follows:
	\begin{equation}
	\ket{\Psi} = \sum_{l,r}C_{l,r}\ket{l}\ket{r},
	\end{equation}
	where $l$ corresponds to the first site $l=j_1$ and $r$ corresponds to the rest $N-1$ sites $r=(j_2,...,j_N)$. Let $C_{l,r}=\sum_a U_{l,a}S_{a,a}V_{a,r}^{\dagger}$ be the singular value decomposition (SVD) of $C_{l,r}$. Then we have
	\begin{equation}\label{eqn:a1}
	\ket{\Psi}= \sum_{l,r}\sum_a U_{l,a}S_{a,a}V_{a,r}^{\dagger}\ket{l}\ket{r}=\sum_aS_a\ket{a}_l\ket{a}_r,
	\end{equation}
	where $\ket{a}_l = \sum_lU_{l,a} \ket{l}$, $\ket{a}_r = \sum_r V_{r,a} \ket{r}$, and $S_{a}=S_{a,a}$. Eqn.~(\ref{eqn:a1}) is called Schmidt decomposition and the entanglement of the bipartite systems $l$ and $r$ is characterized by $S_a$. {Specifically, the bond dimensions between the first site $j_1$ and the rest $N-1$ sites $\{j_2,...,j_N\}$ {are} evaluated by the the number of non-zero values in $S_a$, where the entanglement of the corresponding bipartite systems closely relates to the bond dimension as explained in Subsection \ref{sub:IIc}}.
	
	Through successively performing SVD along each single site in turn, we can split out the rank-$N$ tensor $C_{j_1j_2...j_N}$ into $N$ local tensors $\{A^{j_i}\}_{i=1}^N$. Mathematically, analogous to the Eqn. (\ref{eqn:a1}), the matrix product state of $\ket{\Psi}$ is defined as
\small
	\begin{eqnarray}\label{eqn:A2}
	\ket{\Psi} = &&\sum_{j_1...j_N}\sum_{a_1} U_{j_1,a_1}S_{a_1,a_1}V_{a_1,(j_2...j_N)}^{\dagger}\ket{j_1}\ket{j_2...j_N}\nonumber\\
	=&& \sum_{j_1...j_N}\sum_{a_1}U_{j_1,a_1}C_{a_1,(j_2...j_N)}\ket{j_1}\ket{j_2...j_N}\nonumber\\
	=&& \sum_{j_1...j_N}\sum_{a_1}A_{a_1}^{j_1}U_{(a_1,j_2),a_2}C_{a_2,(j_3...j_N)}\ket{j_1}\ket{j_2}\ket{j_3...j_N}\nonumber\\
	=&&\sum_{j_1...j_N}\sum_{a_1...a_{N-1}}A_{a_1}^{j_1}A_{a_1,a_2}^{j_2}...A_{a_{N-1}}^{j_N}\prod_{i=1}^N\ket{j_i}~,
	\end{eqnarray}
\normalsize	where $S_{a_1,a_1}$ and $V_{a_1,(j_2,...,j_N)}^{\dagger}$ have been multiplied and reshaped to a vector $C_{a_1,(j_2,...,j_N)}$, and the matrix $U_{j_1}$ is decomposed into a collection of $d$ row vectors $A^{j_1}$ with entries $A_{a_1}^{j_1} = U_{j_1,a_1}$. The number of parameters {(elements)} {in MPS} scales as $O(NdM^2)$, where $M$ represents the maximum of bond dimensions among all $S_{a,a}$. When $M$ is small or some truncated methods are employed to keep $M$ small, MPS can efficiently approximate the quantum states with polynomial parameters.

	The String Bond States (SBS) \cite{schuch2008simulation} can be treated as an extension of MPS. The mathematical representation of SBS is
	\begin{equation}\label{eqn:SBS}
	\ket{\Psi} = \prod_{s\in\mathcal{S}}\left(\sum_{a_i,j_1...j_N}\prod_{j_i\in s}A_{a_i}^{j_i,s}\ket{j_1...j_N}\right)~,
	\end{equation}
	where $\mathcal{S}$ is a set of strings, $s\in\mathcal{S}$ is an ordered subset of $\{1,2,...,N\}$, and $A_{a_i}^{j_i,s}$ corresponds to the $A_{a_{i-1},a_i}^{j_i}$ in Eqn. (\ref{eqn:A2}) given {a} fixed $s$. The key idea of SBS is to  place strings of operators on a lattice with $N$ sites. Some examples of string operators is illustrated in Figure \ref{fig:SBS}, where each string operator is denoted by a specific color, i.e., {Figures} \ref{fig:SBS} (a) and (b) have $8$ and $2$ string operators, respectively.
	\begin{figure}[!ht]
		\centering
		\includegraphics[width=3in,height=1.4in]{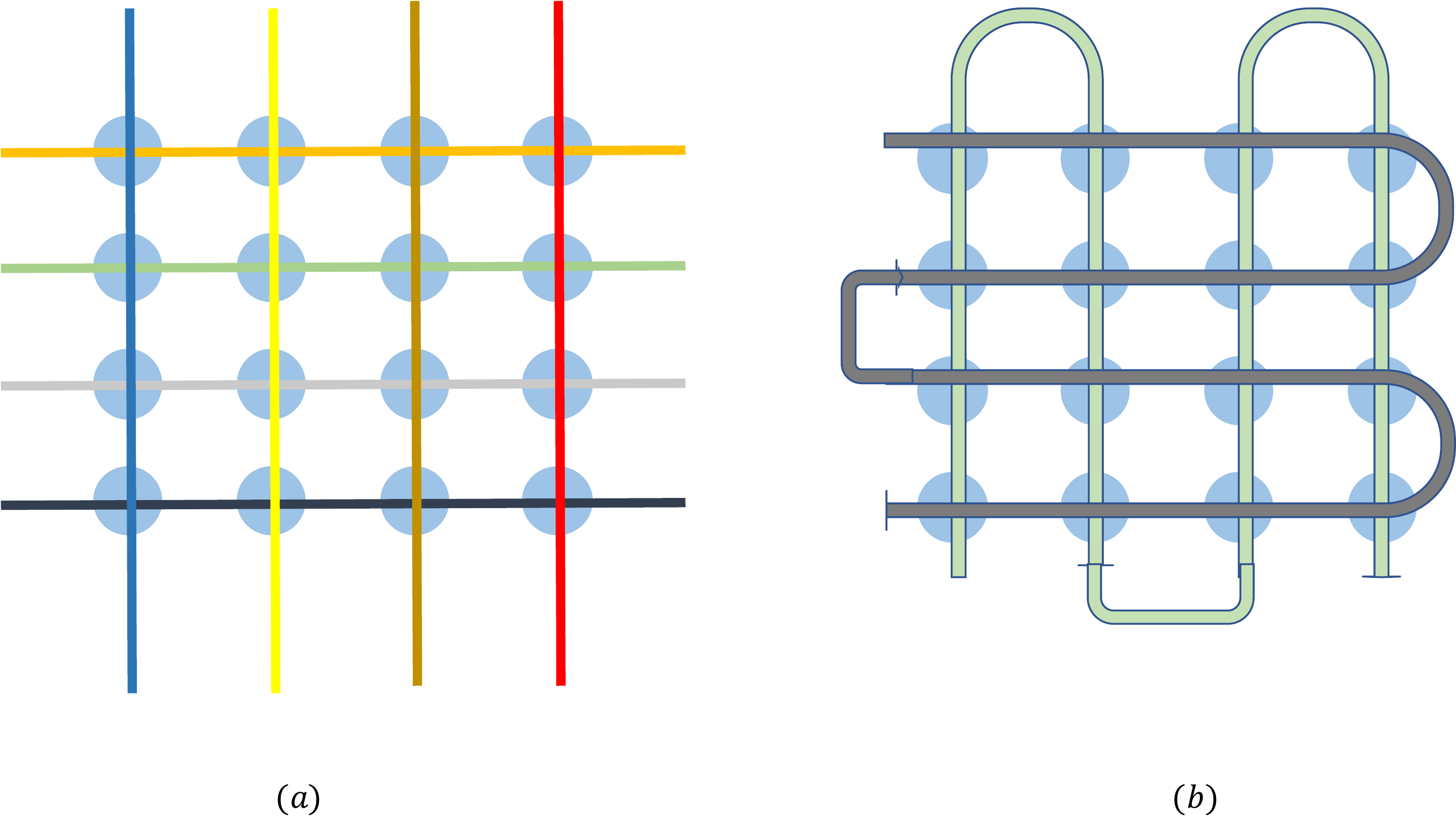}
		\caption{Two examples of string bond states.}
		\label{fig:SBS}
	\end{figure}
	
	\subsection{Entanglement Entropy}\label{sub:IIc}
	The entanglement (also called von Neumann entropy) $\mathcal{S}(\rho)$ of a bipartite system $\rho_{AB}$ is defined as
	\begin{equation}
	\mathcal{S}(\rho_A) = -\text{Tr}(\rho_{A}\ln{\rho_{A}})=-\text{Tr}(\rho_{B}\ln{\rho_{B}})=\mathcal{S}(\rho_B),
	\end{equation}
	where $\rho_{A}= \text{Tr}_{B}\rho_{AB}$ is the reduced density matrix of system $A$. For a quantum system $A$ that satisfies area (volume) law, its entanglement entropy grows proportionally with the boundary area (volume) of system $A$, denoted as $\mathcal{S}(\rho_A)=O(|\partial A|)$ ($=O(|A|)$).
	
{The maximum entanglement entropy of a bipartite system is logarithmically bounded by the bond dimension $D$ as defined in Subsection \ref{sec:2B}, i.e., $\mathcal{S}(\rho_A)\sim\ln D$.} A quantum system $A$ that satisfies area law has an efficient MPS representation, {since in one dimensional case a constant $O(|\partial A|)$ implies  $D$ is also a constant and the number of parameters used in MPS is small.} {On the contrary, a quantum system $A$ that satisfies the volume low implies the entanglement entropy scales with the number of sites, i.e., $\mathcal{S}(\rho_A)\sim N$. Due to $\mathcal{S}(\rho_A)\sim \ln D$, the bond dimension $D$ is scaled as $O(D^N)$, the required parameters in MPS is exponentially large and the quantum system $A$ cannot be efficiently represented by MPS.}
	
	\subsection{Quantum Circuits}
	Analogous to classical computers, a quantum computer accomplishes its computation by  applying quantum gates to quantum bits (qubits).
	
	As stated in \cite{nielsen2010quantum}, a set of single and two qubits gates, which consists of rotation gates and controlled-Not (CNOT) gates, is universal for quantum computation. In other words, any function computable in this model can be computed only using these gates.
	We denote the phase rotation gate $R_{\phi}$, z-axis rotation gate $R_Z(\theta)$, x-axis rotation gate $R_X(\gamma)$, and y-axis rotation gate $R_Y(\alpha)$ as follows: 
	\small
	\begin{alignat*}{2}
	&R_{\phi}=\begin{pmatrix}1 & 0\\ 0 & e^{i\phi}\end{pmatrix}, && R_X(\gamma) = \begin{pmatrix} \cos(\gamma/2) & i\sin(\gamma/2))\\ i\sin(\gamma/2) & \cos(\gamma/2)\end{pmatrix},\\
	&R_Z(\theta)=\begin{pmatrix} e^{i\theta/2} & 0\\ 0 & e^{-i\theta/2}\end{pmatrix},
	&& R_Y(\alpha)=\begin{pmatrix} \cos(\alpha/2) & \sin(\alpha/2))\\ -\sin(\alpha/2) & \cos(\alpha/2)\end{pmatrix}.
	\end{alignat*}
	\normalsize
	
	The CNOT gate, defined as \[\text{CNOT} = \begin{pmatrix} 1 & 0 &0 &0\\0 & 1 &0 &0\\ 0 & 0 &0 &1\\0 & 0 &1 &0 \end{pmatrix},\] flips the target qubit iff the the control qubit is $\ket{1}$. Other quantum gates can be represented by the above universal gate set, e.g., the Pauli-Z gate, defined as $Z=\begin{psmallmatrix} 1 & 0\\ 0 & -1\end{psmallmatrix}$, can be represented by $R_Z(\theta = \pi)$, the T gate, defined as $T=\begin{psmallmatrix} 1 & 0\\ 0 & e^{i\pi/4}\end{psmallmatrix}$, can be represented by $R_{\phi}(\phi = \pi/4)$, the Hadamard gate (H gate), defined as $H =1/\sqrt{2} \begin{psmallmatrix} 1 & 1\\1 & -1\end{psmallmatrix}$, can be represented by $R_X(\gamma=\pi/2)R_Z(\theta=\pi/2)R_X(\gamma=\pi/2)$, {and the two-qubit Controlled-$Z$ gate (shorted as $CZ$ gate), defined as 
		\begin{equation}
		CZ=\begin{pmatrix} 1 & 0 & 0 &0\\ 0 & 1 & 0 &0\\ 0 & 0 & 1 &0\\ 0 & 0 & 0 &-1\\\end{pmatrix},
		\end{equation} 
		can be represented by $(\mathbb{I}\otimes H)\text{CNOT}(\mathbb{I}\otimes H)$}.
	
	Proposition \ref{prop:1} below  demonstrates how to use the universal gate set to express other quantum gates \cite{barenco1995elementary}.
	\begin{prop}\label{prop:1}
A controlled unitary $W$ gate ($CW$) can be simulated by a quantum network {composed of single qubit gates and CNOT gate. Suppose that $W=R_Z(\theta)R_Y(\alpha)R_Z(\beta)$, then as shown in Figure \ref{fig:Ctrl_W}, it can be simulated by the quantum circuits $A, B$ and $C$, where $A=R_Z(\theta)R_Y(\alpha/2)$, $B=R_Y(-\alpha/2)R_Z(-\theta/2-\beta/2)$, and $C=R_Z(\beta/2-\theta/2)$. }  
		\begin{figure}[!ht]
			\begin{tikzcd}
				& \ctrl{1} & \qw \\
				& \gate[style={fill=yellow!20}]{W} & \qw
			\end{tikzcd}
			=	\begin{tikzcd}
				&\qw& \ctrl{1} & \qw & \ctrl{1} &\qw \\
				&\gate[style={fill=yellow!20}]{A}& \gate[style={fill=yellow!20}]{X} & \gate[style={fill=yellow!20}]{B} &\gate[style={fill=yellow!20}]{X} & \gate[style={fill=yellow!20}]{C}
			\end{tikzcd}~,
			\caption{Simulation of controlled unitary gates.}
			\label{fig:Ctrl_W}
		\end{figure}
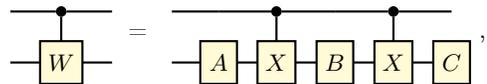
	\end{prop}

	\subsubsection{IQP circuits}
	The instantaneous quantum polynomial (IQP) circuit consists of commute gates that are diagonal in the $Z$ basis. The basic framework of IQP circuits is illustrated in Fig \ref{fig:IQP}.
	\begin{figure}[!ht]
		\centering
		\begin{tikzcd}[row sep={0.8cm,between origins}]
			\lstick{$\ket{0}$} &
			\gate[style={fill=yellow!20}]{H}& \gate[style={fill=green!20},wires=4]{U_Z} & \gate[style={fill=yellow!20}]{H} & \meter{}\\
			\lstick{$\ket{0}$}&\gate[style={fill=yellow!20}]{H} & \hphantom{U_Z very long}& \gate[style={fill=yellow!20}]{H} & \meter{}\\
			\lstick{\vdots}&\gate[style={fill=yellow!20}]{H}&\qwbundle\qw &\gate[style={fill=yellow!20}]{H}&\qw\vdots \\
			\lstick{$\ket{0}$}&\gate[style={fill=yellow!20}]{H} & & \gate[style={fill=yellow!20}]{H} &\meter{}                     	
		\end{tikzcd}
		\caption{A general framework of IQP circuits.}
		\label{fig:IQP}
	\end{figure}
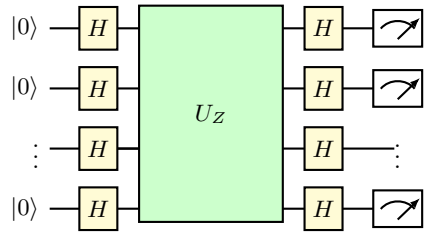
	
	Given $N$ qubits, the IQP circuits can generate distributions $p_I = \left|\braket{0^{\otimes N}|H^{\otimes N}U_ZH^{\otimes N}|0^{\otimes N}}\right|^2$, where $U_Z$ is composed of $O(poly(N))$ commuting gates, e.g., the single-qubit $T$ gate and {$CZ$ gate}.
	
	IQP circuits are proven to be capable of  generating probability distributions $p_I$ that cannot be classically simulated efficiently \cite{bremner2010classical}.  The main result of IQP is summarized in the following proposition.
	\begin{prop}\label{prop1}
		If the output probability distributions generated by uniform families of IQP circuits could be weakly classically simulated to within multiplicative error $1\leq c\leq \sqrt{2}$, then $\textit{post-BPP} = PP$ and the $PH$ would collapse to its third level.
	\end{prop}
	
	\subsection{Parameterized Quantum Circuits}	\label{sec:IIE}
	Parameterized quantum circuits (PQCs), as a special type of quantum circuit model, are composed of a set of parameterized single and controlled single qubit gates. In this work, a PQC is used to implement a unitary transformation operator $U(\bm{\theta})$ with $O(poly (N))$ parameterized quantum gates, where $N$ is the number of input qubits.
	
	Several recent works \cite{benedetti2018generative,liu2018differentiable,huggins2018towards} have employed PQCs to accomplish generative tasks. {One major reason is that} the superposition property allows the number of trainable parameters to be dramatically reduced. In generative tasks, PQCs produce  the probability $q(X=\bm{x}) = |\braket{\bm{x}|\Psi_G}|^2$ measured by the computational basis $\ket{\bm{x}}$, where
	\begin{equation}\label{eqn:sp0}
	\ket{\Psi_G} = U(\bm{\theta})\ket{0}^{\otimes N}.
	\end{equation}
	The parameters $\bm{\theta}$ can be  optimized using only classical approaches,
	\begin{equation}\label{eqn:minloss}
	\arg \min_{\bm{\theta}} \mathcal{L}(q(X),p(X)),
	\end{equation}
	where  $\mathcal{L}(\cdot,\cdot)$ is a loss function that measures the dissimilarity of  the  generated and the targeted probability distributions. For example, suppose that the loss function is negative log-likelihood \cite{platt1999probabilistic}, the optimizing process is,
	\begin{equation}\label{eqn:sp2}
	\arg \min_{\bm{\theta}} \frac{1}{D}\sum_{i=1}^D -\log q(X=\bm{x}_i), X\sim p(X),
	\end{equation}
	where the dataset $\mathcal{D}=\{\bm{x}_i\}_{i=1}^D$ is sampled from the targeted probability distribution $p(X)$, the size of $\mathcal{D}$ is $D$, and each example of $\mathcal{D}$ is denoted as $\bm{x}_i$ for $i\in[1,D]$.
	
{Another loss function that is broadly employed is the maximum mean discrepancy (MMD). The MMD loss is defined as \begin{equation}\label{eqn:9}
\mathcal{L} = \left\|\sum_{\bm{\lambda}}\sum_{x^i\in\bm{x}}q(x^i,\bm{\lambda})\phi(x^i)-\sum_{x^i\in\bm{x}}p{(x^i)}\phi(x^i)\right\|^2~,
\end{equation}
where $\phi(x^i)$ maps the $i$-th input data, $x^i$, into a high-dimensional  Reproducing Kernel Hilbert Space \cite{berlinet2011reproducing}, and $\sum_{x^i\in\bm{x}}p(x^i)$ refers to the target probability distribution. More details about the MMD loss and how to optimize it by employing the gradient descent method with unbiased estimation are introduced in \cite{liu2018differentiable,mitarai2018quantum}.}
	
	In the following, we define two types of PQCs that are the focus of the present work, where the major difference is the layout of quantum gates to compose $U(\bm{\theta})$.
	
	\subsubsection{Multilayer Parameterized Quantum Circuits}
	Multilayer Parameterized Quantum Circuits (MPQCs) are composed of $L$ blocks, where each block implements $U(\bm{\theta}^{i})$, with $i\in[1, L]$ and $L\sim poly(N)$. A unitary operator $U(\bm{\theta})=\prod_{i=1}^LU(\bm{\theta}^i)$ is applied to $N$ input qubits. An example of MPQC is illustrated in Fig.~\ref{fig:MPQC00}. 
	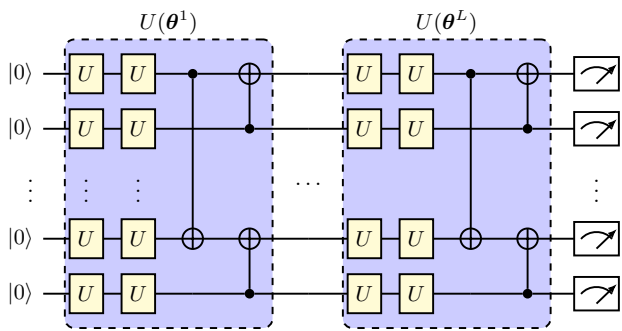
\begin{figure}[!ht]
		\begin{adjustbox}{width=0.47\textwidth}
			\begin{tikzcd}[row sep={0.8cm,between origins}]
				\lstick{$\ket{0}$} &\gategroup[wires=5,steps=4,style={dashed,rounded corners,fill=blue!20, inner xsep=2pt},background]{$U(\bm{\theta}^1)$}
				\gate[style={fill=yellow!20}]{U}& \gate[style={fill=yellow!20}]{U} & \ctrl{3} &\targ{}&\qw &\gategroup[wires=5,steps=4,style={dashed,rounded corners,fill=blue!20, inner xsep=2pt},background]{$U(\bm{\theta}^L)$}\gate[style={fill=yellow!20}]{U}&\gate[style={fill=yellow!20}]{U}
				& \ctrl{3} &\targ{} &\meter{}\\
				\lstick{$\ket{0}$}&\gate[style={fill=yellow!20}]{U} & \gate[style={fill=yellow!20}]{U}&\qw  & \ctrl{-1}&\qw
				&\gate[style={fill=yellow!20}]{U}&\gate[style={fill=yellow!20}]{U}
				&\qw  & \ctrl{-1}
				&\meter{}\\
				\lstick{\vdots}&\vdots &\vdots & & & \ldots & & & &&\vdots\\ 
				\lstick{$\ket{0}$}&\gate[style={fill=yellow!20}]{U}&\gate[style={fill=yellow!20}]{U} &\targ{}&\targ{}&\qw &\gate[style={fill=yellow!20}]{U}&\gate[style={fill=yellow!20}]{U}
				&\targ{}&\targ{}
				&\meter{} \\
				\lstick{$\ket{0}$}&\gate[style={fill=yellow!20}]{U} &\gate[style={fill=yellow!20}]{U} & \qw &\ctrl{-1} &\qw &\gate[style={fill=yellow!20}]{U}&\gate[style={fill=yellow!20}]{U}
				& \qw &\ctrl{-1}  &\meter{}                     	
			\end{tikzcd}
		\end{adjustbox}
		\caption{Illustration of MPQCs.}
		\label{fig:MPQC00}
	\end{figure}
	In each block, the arrangement of quantum gates is identical. Moreover, each qubit is operated with at least one parameterized gate (denoted by yellow color), and  CNOT gates within the block can connect arbitrary two qubits. Another requirement in MPQCs is,  the amount of CNOT gates is no larger than $N$ in each block.
	
	Using MPQCs to accomplish generative tasks have been explored by \cite{benedetti2018generative,liu2018differentiable}, while the layout of quantum gates in each block and the optimization methods are varied.
	
	\subsubsection{Tensor Network Parameterized Quantum Circuits}
	Another type of PQCs is the tensor network PQCs (TPQCs), which generally inherit the tensor network structures, i.e., MPS, tree tensor network\tl{,} or MERA. In other words, CNOT gates can only connect two local qubits. Mathematically, the quantum state $\ket{\Psi_G}$ generated by TPQCs is formulated as
	\begin{equation}
	\ket{\Psi_G} = \prod_{i=1}^{L}\bigotimes_{j=1}^{M_i}U(\bm{\theta}^i_j)\ket{0}^{\otimes N},
	\end{equation}
	where $M_i$ represents the number of local blocks in the block $U(\bm{\theta}^i)$. For example, a TPQC that inherits from the layout of tree tensor network is given in Fig.~\ref{fig:TPQC1}.  
	\begin{figure}[!ht]
		\begin{adjustbox}{width=0.48\textwidth}
			\begin{tikzcd}[row sep={0.8cm,between origins}]	
				\lstick{$\ket{0}$} &\gategroup[wires=2,steps=4,style={dashed,rounded corners,fill=blue!20, inner xsep=2pt},background]{}\gategroup[wires=9,steps=4,style={dashed,rounded corners, inner xsep=2pt},background]{$U(\bm{\theta}^1)$}\gate[style={fill=yellow!20}]{U}&\ctrl{1}&\gate[style={fill=yellow!20}]{U} &\ctrl{1} &\gategroup[wires=9,steps=4,style={dashed,rounded corners,fill=blue!20, inner xsep=2pt},background]{$U(\bm{\theta}^2)$}\qw &\qw &\qw &\qw &\qw &\gategroup[wires=9,steps=4,style={dashed,rounded corners,fill=blue!20, inner xsep=2pt},background]{$U(\bm{\theta}^L)$}\qw &\qw &\qw &\qw &\meter{}\\
				\lstick{$\ket{0}$} &\gate[style={fill=yellow!20}]{U} &\targ{}&\gate[style={fill=yellow!20}]{U}&\targ{}&\gate[style={fill=yellow!20}]{U} &\ctrl{2} &\gate[style={fill=yellow!20}]{U} &\ctrl{2} &\qw &\qw &\qw &\qw &\qw &\meter{}\\
				\lstick{$\ket{0}$} &\gategroup[wires=2,steps=4,style={dashed,rounded corners,fill=pink!20, inner xsep=2pt},background]{}\gate[style={fill=yellow!20}]{U} &\ctrl{1} &\gate[style={fill=yellow!20}]{U} &\ctrl{1} &\qw &\qw &\qw &\qw &\qw &\qw&\qw &\qw &\qw &\meter{}\\
				\lstick{$\ket{0}$} &\gate[style={fill=yellow!20}]{U} &\targ{} &\gate[style={fill=yellow!20}]{U} &\targ{} &\gate[style={fill=yellow!20}]{U} &\targ{} &\gate[style={fill=yellow!20}]{U} &\targ{} &\qw &\gate[style={fill=yellow!20}]{U}&\ctrl{5}&\gate[style={fill=yellow!20}]{U}&\ctrl{5}&\meter{}\\
				\lstick{\vdots}  & &  & \vdots & & & & & &\ldots & & & & &\vdots\\ 
				\lstick{$\ket{0}$} &\gategroup[wires=2,steps=4,style={dashed,rounded corners,fill=yellow!20, inner xsep=2pt},background]{}\gate[style={fill=yellow!20}]{U} &\ctrl{1} &\gate[style={fill=yellow!20}]{U} &\ctrl{1}	&\qw &\qw &\qw &\qw &\qw &\qw &\qw &\qw &\qw &\meter{} \\
				\lstick{$\ket{0}$} &\gate[style={fill=yellow!20}]{U} &\targ{} &\gate[style={fill=yellow!20}]{U} &\targ{} &\gate[style={fill=yellow!20}]{U}&\ctrl{2} &\gate[style={fill=yellow!20}]{U} &\ctrl{2} &\qw &\qw &\qw &\qw &\qw &\meter{} \\
				\lstick{$\ket{0}$} &\gategroup[wires=2,steps=4,style={dashed,rounded corners,fill=green!20, inner xsep=2pt},background]{}\gate[style={fill=yellow!20}]{U} &\ctrl{1} &\gate[style={fill=yellow!20}]{U} &\ctrl{1}	&\qw &\qw &\qw &\qw &\qw &\qw &\qw &\qw &\qw  &\meter{} \\
				\lstick{$\ket{0}$} &\gate[style={fill=yellow!20}]{U} &\targ{} &\gate[style={fill=yellow!20}]{U} &\targ{} &\gate[style={fill=yellow!20}]{U}&\targ{} &\gate[style={fill=yellow!20}]{U} &\targ{} &\qw &\gate[style={fill=yellow!20}]{U}&\targ{} &\gate[style={fill=yellow!20}]{U} &\targ{}
				&\meter{} 
			\end{tikzcd}
		\end{adjustbox}	
		\caption{An Example of TPQCs where the CNOT gates in different layers has different local constraints.}
		\label{fig:TPQC1}
	\end{figure}
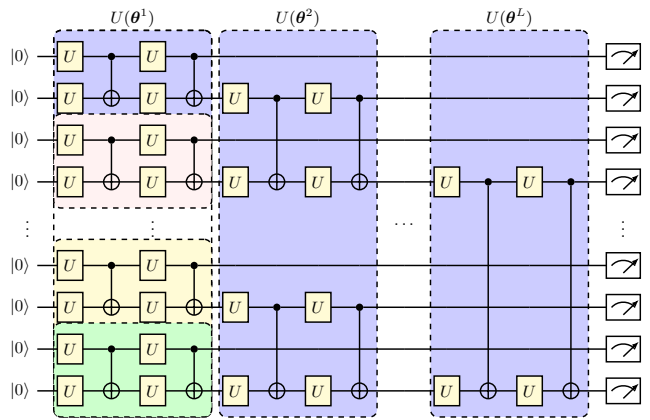
	Figure~\ref{fig:TPQC2} illustrates another example of TPQC.
	Employing TPQCs to accomplish generative tasks has been investigated in \cite{huggins2018towards}.
	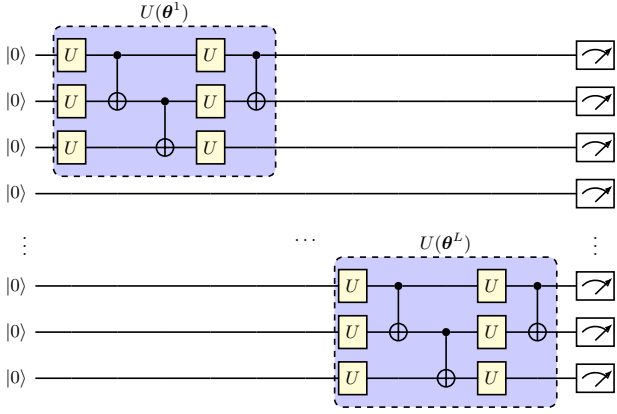
\begin{figure}[!ht]
	\begin{adjustbox}{width=0.47\textwidth}\label{sss}
		\begin{tikzcd}[row sep={0.8cm,between origins}]
			\lstick{$\ket{0}$} &\gategroup[wires=3,steps=5,style={dashed,rounded corners,fill=blue!20, inner xsep=2pt},background]{$U(\bm{\theta}^1)$}\gate[style={fill=yellow!20}]{U}&\ctrl{1}&\qw &\gate[style={fill=yellow!20}]{U}&\ctrl{1}  &\qw &\qw &\qw &\qw&\qw&\qw&\meter{}\\
			\lstick{$\ket{0}$} &\gate[style={fill=yellow!20}]{U} &\targ{}&\ctrl{1}  &\gate[style={fill=yellow!20}]{U}&\targ{}&\qw &\qw &\qw &\qw &\qw &\qw &\meter{}\\
			\lstick{$\ket{0}$} &\gate[style={fill=yellow!20}]{U} &\qw &\targ{} &\gate[style={fill=yellow!20}]{U} &\qw  &\qw &\qw &\qw &\qw &\qw &\qw &\meter{}\\
			\lstick{$\ket{0}$} &\qw &\qw &\qw &\qw&\qw &\qw &\qw &\qw &\qw &\qw &\qw&\meter{}\\
			\lstick{\vdots}    & & & & & &\ldots & & & & & &\vdots\\ 
			\lstick{$\ket{0}$} &\qw &\qw  &\qw &\qw &\qw &\qw &\gategroup[wires=3,steps=5,style={dashed,rounded corners,fill=blue!20, inner xsep=2pt},background]{$U(\bm{\theta}^L)$}\gate[style={fill=yellow!20}]{U} &\ctrl{1} &\qw &\gate[style={fill=yellow!20}]{U} &\ctrl{1}&\meter{} \\
			\lstick{$\ket{0}$}  &\qw &\qw &\qw &\qw &\qw &\qw &\gate[style={fill=yellow!20}]{U} &\targ{} &\ctrl{1} &\gate[style={fill=yellow!20}]{U}&\targ{}&\meter{} \\
			\lstick{$\ket{0}$}&\qw &\qw &\qw &\qw &\qw &\qw &\gate[style={fill=yellow!20}]{U} &\qw &\targ{} &\gate[style={fill=yellow!20}]{U}&\qw &\meter{}  
		\end{tikzcd}
	\end{adjustbox}	
	\caption{An Example of TPQCs that inherits the layout of MPS.}
	\label{fig:TPQC2}
\end{figure}

	\section{Expressive Power Parameterized Quantum Circuits}
	The goal of a generative learning network is to learn a distribution $q(x)$ that approximates a targeted probability distribution $p(x)$ within a tolerable error $\epsilon$. The expressive power of a generative learning machine directly determines how well the generated distribution can match the target distribution (e.g. Eqn. (\ref{eqn:minloss})). The stronger the expressive power is, the smaller  the dissimilarity of two distributions will be. 
	
	One of the main results in this paper is as follows.
	\begin{thm}\label{thmv1}
		The expressive power of MPQCs and TPQCs with $O(poly(N))$ single qubits gates and CNOT gates, and classical neural networks with $O(poly(N))$ trainable parameters, where $N$ refers to the number of qubits or the visible units,  can be ordered as:  MPQCs $ > $ DBM $>$ long range RBM $>$TPQCs $>$ short range RBM.
	\end{thm}
	\begin{proof}
		
		This theorem can be proved following Theorems \ref{thm1}, \ref{lemm2} and \ref{thm2}.  Theorem~\ref{thm1} below demonstrates that MPQCs and TPQCs are capable of simulating quantum systems with volume law entanglement. Combined with the results in \cite{gao2017efficient1,chen2018equivalence,deng2017quantum},  DBM and the long range RBM are also capable of efficiently representing quantum states with volume law entanglement, whereas the short range RBM can only efficiently represent quantum states with area law entanglement.  {Theorem~\ref{lemm2} proves that some distributions which can be efficiently generated by MPQCs, DBM, and long range RBM, are difficult to be generated by TPQCs. Next, by making connections with IQP circuits in Theorem~\ref{thm:5}, we can further prove: $\text{MPQCs} > \text{DBM} $. Since it has been proved that DBM has a stronger expressive power than long range RBM \cite{gao2017efficient1},} this concludes the theorem. 
	\end{proof}

	We first relate MPQCs and TPQCs to tensor network states in the following theorem.
	\begin{thm}\label{thm1}
		MPS with bond dimension $D$ can be efficiently represented by MPQCs and TPQCs 
		with $O(poly(\log{D}))$ blocks, where each block contains $O(N)$ trainable parameters and at most $N$ CNOT gates.
	\end{thm}
	\begin{proof}
		In PQCs, only CNOT gates can increase bond dimensions. Specifically,
		\begin{equation}
		\text{CNOT} = \sum_{\sigma,\tau,\sigma', \tau' \in\{0,1\}} \mathcal{O}_{\sigma,\tau}^{\sigma',\tau'}\ket{\sigma'}\ket{\tau'}\bra{\tau}\bra{\sigma},
		\end{equation}
		where $\mathcal{O}_{\sigma,\tau}^{\sigma',\tau'}$ is a rank-$4$ tensor with $\mathcal{O}_{\sigma=0,\tau=0}^{\sigma'=0,\tau'=0}=\mathcal{O}_{\sigma=1,\tau=1}^{\sigma'=1,\tau'=1}=\mathcal{O}_{\sigma=0,\tau=1}^{\sigma'=1,\tau'=1}=\mathcal{O}_{\sigma=1,\tau=1}^{\sigma'=1,\tau'=0}=1$ and  $0$ otherwise. The CNOT gate can be decomposed into two local tensors with bond dimension $D=2$. One possible solution is
		\begin{equation}\label{eqn:CNOT_decompose}
		\mathcal{O}_{\sigma,\tau}^{\sigma',\tau'} = \sum_{b\in\{0,1\}}W^{\sigma',\sigma}_{1b}W^{\tau',\tau}_{2b}~,
		\end{equation}
		where $W^{\sigma',\sigma}_{1b}$ and $W^{\tau',\tau}_{2b}$ correspond to two local rank-$3$ tensors, and their explicit representations are as follows:
		\begin{eqnarray}\label{eqn:CNOT}
		&&W^{\sigma',\sigma}_{10} = \begin{bmatrix}
		1 & 0 \\
		0& 0
		\end{bmatrix},
		W^{\sigma',\sigma}_{11} = \begin{bmatrix}
		1 & 0 \nonumber\\
		0& 1
		\end{bmatrix},\\
		&&	W^{\tau',\tau}_{20} = \begin{bmatrix}
		1 & -1 \\
		-1& 1
		\end{bmatrix},
		W^{\tau',\tau}_{21} = \begin{bmatrix}
		0 & 1 \\
		1& 0
		\end{bmatrix}~.
		\end{eqnarray}
		
		{Suppose that there exists $k$ CNOT gates between the $i$-th and $(i+1)$-th qubits, where the first $i$ qubits and the remaining $N-i$ qubits compose a bipartite system, the maximal bond dimension of such a bipartite system is  $2^k$.  
			Since the bond dimension exponentially scales with the number of CNOT gates, 
			$O(poly(\log D))$ blocks are required to generate an MPS with bond dimension $D$.}
		
	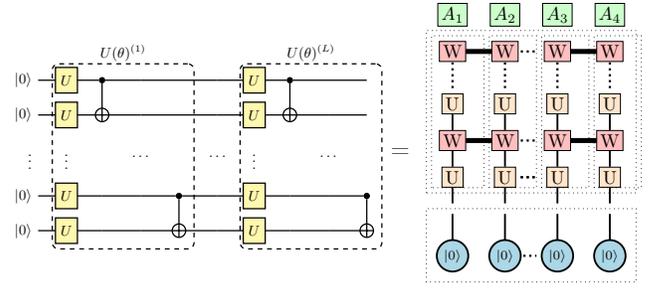
\begin{figure}[!ht]
	\begin{adjustbox}{width=0.28\textwidth}
		\begin{tikzcd}
			\lstick{$\ket{0}$} &\gategroup[wires=5,steps=4,style={dashed,rounded corners,inner xsep=2pt},background]{$U(\theta)^{(1)}$}\gate[style={fill=yellow!40}]{U}  &\ctrl{1} &\qw &\qw &\qw  &\gategroup[wires=5,steps=4,style={dashed,rounded corners,inner xsep=2pt},background]{$U(\theta)^{(L)}$}\gate[style={fill=yellow!40}]{U}  &\ctrl{1} &\qw &\qw\\
			\lstick{$\ket{0}$} &\gate[style={fill=yellow!40}]{U}  & \targ{} &\qw  &\qw &\qw &\gate[style={fill=yellow!40}]{U}  & \targ{} &\qw  &\qw \\
			\lstick{\vdots}  &\vdots & &\ldots & &\ldots &\vdots & &\ldots &\\
			\lstick{$\ket{0}$} &\gate[style={fill=yellow!40}]{U} & \qw  &\qw &\ctrl{1} & \qw &\gate[style={fill=yellow!40}]{U} & \qw  &\qw &\ctrl{1} \\
			\lstick{$\ket{0}$} &\gate[style={fill=yellow!40}]{U} & \qw  &\qw &\targ{} & \qw &\gate[style={fill=yellow!40}]{U} & \qw  &\qw &\targ{}	
		\end{tikzcd}
	\end{adjustbox}
	=
	\raisebox{-.45\height}{
		\begin{adjustbox}{width=0.16\textwidth}
			\begin{tikzpicture}
			\draw[dotted] (-.5,-.5) rectangle (3.5,0.9);
			\Vertex[label=\ket{0},Math]{A}
			\draw[line  width=1pt] (0,0.3) -- (0,.8);
			\Vertex[x=1,label=\ket{0},Math]{B}
			\draw[line  width=1pt] (1,.3) -- (1,.8);
			\draw [very thick, dotted] (1.35, 0) -- (1.6,0);
			\Vertex[x=2,label=\ket{0},Math]{C}
			\draw[line  width=1pt] (2,.3) -- (2,.8);
			\Vertex[x=3,label=\ket{0},Math]{D}
			\draw[line  width=1pt] (3,.3) -- (3,.8);
			\draw [line  width=1pt] (0,1) -- (0,1.3);
			\draw [line  width=1pt] (1,1) -- (1,1.3);
			\draw [line  width=1pt] (2,1) -- (2,1.3);
			\draw [line  width=1pt] (3,1) -- (3,1.3);
			\Text[x=1,y=1.5,style={draw, rectangle,fill=orange!20}]{U}
			\Text[x=2,y=1.5,style={draw, rectangle,fill=orange!20}]{U}
			\draw [very thick, dotted] (1.3, 1.5) -- (1.6,1.5);
			\Text[x=3,y=1.5,style={draw, rectangle,fill=orange!20}]{U}
			\Text[x=0,y=1.5,style={draw, rectangle,fill=orange!20}]{U}
			\draw [line  width=1pt] (0,1.7) -- (0,2.0);
			\draw [line  width=1pt] (1,1.7) -- (1,2.0);
			\draw [line  width=1pt] (2,1.7) -- (2,2.0);
			\draw [line  width=1pt] (3,1.7) -- (3,2.0);
			\Text[x=1,y=2.2,style={draw, rectangle,fill=pink}]{W}
			\draw [line,line width=0.95mm] (.2, 2.2) -- (.75,2.2);
			\Text[x=2,y=2.2,style={draw, rectangle,fill=pink}]{W}
			\draw [very thick, dotted] (1.3, 2.2) -- (1.6,2.2);
			\Text[x=3,y=2.2,style={draw, rectangle,fill=pink}]{W}
			\draw [line,line width=0.95mm] (2.25, 2.2) -- (2.75,2.2);
			\Text[x=0,y=2.2,style={draw, rectangle,fill=pink}]{W}
			\draw [line  width=1pt] (0,2.4) -- (0,2.7);
			\draw [line  width=1pt] (1,2.4) -- (1,2.7);
			\draw [line  width=1pt] (2,2.4) -- (2,2.7);
			\draw [line  width=1pt] (3,2.4) -- (3,2.7);
			\Text[x=1,y=2.9,style={draw, rectangle,fill=orange!20}]{U}
			\Text[x=2,y=2.9,style={draw, rectangle,fill=orange!20}]{U}
			\draw [very thick, dotted] (1.3, 1.5) -- (1.6,1.5);
			\Text[x=3,y=2.9,style={draw, rectangle,fill=orange!20}]{U}
			\Text[x=0,y=2.9,style={draw, rectangle,fill=orange!20}]{U}
			\draw [very thick, dotted] (0,3.2) -- (0,3.7);
			\draw [very thick, dotted](1,3.2) -- (1,3.7);
			\draw [very thick, dotted] (2,3.2) -- (2,3.7);
			\draw [very thick, dotted] (3,3.2) -- (3,3.7);
			\Text[x=1,y=3.9,style={draw, rectangle,fill=pink}]{W}
			\draw [line,line width=0.95mm] (.2, 3.9) -- (.75,3.9);
			\Text[x=2,y=3.9,style={draw, rectangle,fill=pink}]{W}
			\draw [very thick, dotted] (1.3, 3.9) -- (1.6,3.9);
			\Text[x=3,y=3.9,style={draw, rectangle,fill=pink}]{W}
			\draw [line,line width=0.95mm] (2.25, 3.9) -- (2.75,3.9);
			\Text[x=0,y=3.9,style={draw, rectangle,fill=pink}]{W}
			\draw[dotted] (-.5,1.2) rectangle (3.5,4.3);
			\draw[dotted] (-.4,1.3) rectangle (0.6,4.2);
			\Text[x=0,y=4.6,style={draw, rectangle,fill=green!20}]{$A_1$}
			\draw[dotted] (.7,1.3) rectangle (1.6,4.2);
			\Text[x=1,y=4.6,style={draw, rectangle,fill=green!20}]{$A_2$};
			\draw[dotted] (1.7,1.3) rectangle (2.6,4.2);
			\Text[x=2,y=4.6,style={draw, rectangle,fill=green!20}]{$A_3$};
			\draw[dotted] (2.7,1.3) rectangle (3.6,4.2);
			\Text[x=3,y=4.6,style={draw, rectangle,fill=green!20}]{$A_4$};
			\end{tikzpicture}
		\end{adjustbox}
	}
	\caption{\small{ The mapping between MPQC and MPS.}}
	\label{fig:5-1}
\end{figure}
		
		Although CNOT increases the bond dimensions, it cannot directly represent arbitrary local tensors $A_{a_i}^{j_i}$ defined in Eqn. (\ref{eqn:A2}), because the local tensors $W$ of CNOT gates defined in Eqn. (\ref{eqn:CNOT_decompose}) are fixed. This issue can be tackled by using the parameterized single qubit gates. In summary, any MPS with bond dimension $D$ can be simulated by PQCs with $O(poly(\log D))$ blocks so that CNOT gates contribute to increase the bond dimensions and parameterized single qubit gates contribute to form arbitrary local tensors.
		
		{Figure \ref{fig:5-1} depicts a mapping between MPQC and MPS, where, for illustrative purpose, we assume that $N-1$ CNOT gates are applied to the data qubits in sequence. The middle section of Figure  \ref{fig:5-1} indicates the effects of CNOT gates and parameterized single qubit gates. All local tensors applied to the same qubit can be merged into one local tensor (c.f. Eqn.~(\ref{eqn:A2})) and yield the corresponding MPS, as shown in the right section of Figure  \ref{fig:5-1}.}
		
		
	\end{proof}
	
	Theorem \ref{thm1} implies that MPQCs and TPQCs with polynomial (logarithmic) blocks can efficiently represent quantum states with volume (area) law entanglement. However, the expressive power does not solely depend on the volume of entanglement alone. Even though both long range RBM and MPQCs can represent quantum states with volume law, some quantum states, such as  those generated by the translation-invariant Ising spin model \cite{gao2017efficient1}, can be efficiently represented by constant depth quantum circuits, but are hard for RBM. %
	
	Two major differences between TPQCs and MPQCs are that (i) CNOT gates in TPQCs cannot connect any two qubits arbitrarily and (ii) the blocks are replicated based on the structure of the tensor networks. This restriction limits the expressive power of TPQCs.
	
	\begin{thm}\label{lemm2}
		Some probability distributions generated by MPQCS, DBM, and long range RBM cannot be efficiently generated by TPQCs.
	\end{thm}
	
	\begin{proof}
		{The theorem is proved by construction. In DBM and long range RBM, the correlation between any two visible units  can be built by linking to the same hidden unit. Similarly, in MPQCs, any two qubits can build correlation by applying a CNOT gate. We provide an example to show that the distribution can be easily generated by DBM, long range RBM, MPQCs but can be difficult for TPQCs. Given $N$ binary inputs $\{v_i\}_{i=1}^N$ where $v_1=1$ and $v_i=0$ for $i\in\{2,3,...,N\}$, we define the targeted distribution as $p(v_1=1, v_i=0, v_N=1)=1$ with $i\in\{2,3,...,N-1\}$. For DBM and long range RBM, this distribution can be generated by introducing one hidden unit $h_1$. As shown in the left panel of Figure \ref{fig:6-1}, each visible unit encodes a binary input and the number of trainable parameters is $2$. Similarly, the distribution can be generated by MPQCs.  By encoding $v_i$ into the $i$-th qubit, only one  CNOT gate is required to connect the first and the $N$-th qubit, as illustrated in the right panel of the Figure \ref{fig:6-1}. However, this distribution cannot be  efficiently generated by TPQCs, which prevents long range interaction.   }
\begin{figure}[!ht]
	\raisebox{-.45\height}{				\begin{adjustbox}{width=0.25\textwidth}
			\begin{tikzpicture}
			\Vertex[label=v_1,Math]{A}	\Vertex[x=1, label=v_2,Math]{B}
			\Vertex[x=2,label=v_3,Math]{C}
			\draw[very thick, dotted] (2.7, 0) -- (3, 0);
			\Vertex[x=4,label=v_N,Math]{D}
			\Vertex[x=2, y=2.5,color=red!30,label=h_1,Math]{E}
			\Edge(A)(E)
			\Edge(D)(E)
			\end{tikzpicture}
		\end{adjustbox}
	}
	\begin{adjustbox}{width=0.20\textwidth}
		\begin{tikzcd}[row sep={0.8cm,between origins}]
			\lstick{$\ket{1}$} &\qw &\gategroup[wires=5,steps=2,style={dashed,rounded corners,inner xsep=2pt},background]{} \qw &\ctrl{4}   &\meter{}\\
			\lstick{$\ket{0}$} &\qw  & \qw &\qw  &\meter{}\\
			\lstick{$\ket{0}$} &\qw & \qw  &\qw  &\meter{}\\
			\lstick{\vdots}  &\vdots &\ & &\vdots\\
			\lstick{$\ket{0}$} &\qw & \qw  &\targ{}  &\meter{}           	
		\end{tikzcd}
	\end{adjustbox}
	\caption{A toy example to demonstrate that a probability distribution cannot be efficiently generated by TPQCs.}
	\label{fig:6-1}
\end{figure}
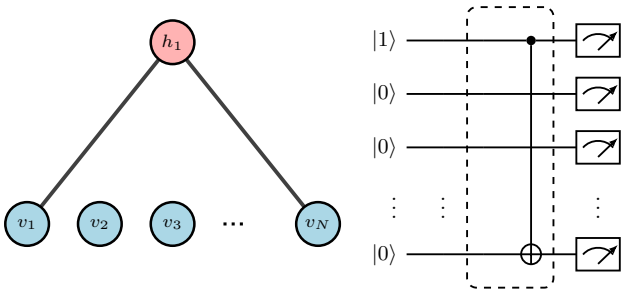
	\end{proof}
	
	%
	
	From the perspective of computational complexity, we can obtain the following theorem:
	\begin{thm}\label{thm2}
		There exist probability distributions generated by MPQCs with $O(poly(N))$ blocks, where $N$ is the number of input quantum states which cannot be simulated efficiently by classical neural networks unless the polynomial  hierarchy (PH) collapses.
	\end{thm}
	\begin{proof}
		This theorem can be proved by combining Theorem \ref{thm:5} below with Proposition \ref{prop1}. Theorem \ref{thm:5} shows that any IQP circuits with $N$ qubits and $O(poly(N))$ commuting gates can be transformed into MPQCs with $O(poly(N))$ blocks.  As stated in proposition \ref{prop1}, there exist probability distributions, generated by IQP, that cannot be efficiently simulated by classical circuits (including DBM or long range RBM).
	\end{proof}
	
	\begin{thm}\label{thm:5}
		MPQCs can efficiently simulate any IQP circuits with $N$ qubits and $O(poly(N))$ commuting gates, with at most $O(poly (N))$ blocks, where each block contains no more than $7N$ single qubit gates and $N-1$ CNOT gates.
	\end{thm}
	\begin{proof}
		
		A general IQP circuit is shown in Figure \ref{fig:general_IQP}.
		\begin{figure}[!ht]
	\begin{adjustbox}{width=0.47\textwidth}
		\begin{tikzcd}[row sep={0.8cm,between origins}]
			\lstick{$\ket{0}$} &\gategroup[wires=7,steps=1,style={dashed,rounded corners, inner xsep=2pt},background]{$H$} \gate[style={fill=yellow!20}]{H}&\gategroup[wires=7,steps=7,style={dashed,rounded corners, inner xsep=2pt},background]{$U_Z$}\gate[style={fill=yellow!20}]{T} &\ctrl{1}	&\gategroup[wires=7,steps=1,style={dashed,rounded corners, fill=blue!15, inner xsep=2pt},background]{}\gate[style={fill=yellow!20}]{T} &\qw&\qw&\ctrl{2}&\qw	&\gategroup[wires=7,steps=1,style={dashed,rounded corners, inner xsep=2pt},background]{$H$} \gate[style={fill=yellow!20}]{H} &\meter{}  \\
			\lstick{$\ket{0}$} &\gate[style={fill=yellow!20}]{H} &\qw &\gate[style={fill=yellow!20}]{Z} &\ctrl{1} &\ctrl{3}	&\qw&\qw&\ctrl{5}&\gate[style={fill=yellow!20}]{H}&\meter{}  	\\
			\lstick{$\ket{0}$} &\gate[style={fill=yellow!20}]{H} &\gate[style={fill=yellow!20}]{T} &\qw &\gate[style={fill=yellow!20}]{Z}&\qw&\qw&\gate[style={fill=yellow!20}]{Z}&\qw	&\gate[style={fill=yellow!20}]{H}&\meter{}  		\\
			\lstick{\vdots}    &\vdots	&	&	&	&	&\ldots	& &	&\vdots	&\vdots  		\\
			\lstick{$\ket{0}$} &\gate[style={fill=yellow!20}]{H} &\qw &\ctrl{2}	&\ctrl{1} &\gate[style={fill=yellow!20}]{Z}&\qw	&\qw&\qw&\gate[style={fill=yellow!20}]{H}	&\meter{}  \\
			\lstick{$\ket{0}$} &\gate[style={fill=yellow!20}]{H} &\gate[style={fill=yellow!20}]{T} &\qw	&\gategroup[wires=1,steps=1,style={dashed,rounded corners, fill=red!20,inner xsep=2pt},background]{}\gate[style={fill=yellow!20}]{Z}&\qw&\qw&\qw&\qw&\gate[style={fill=yellow!20}]{H}	&\meter{}  	\\
			\lstick{$\ket{0}$} &\gate[style={fill=yellow!20}]{H}  &\qw  &\gate[style={fill=yellow!20}]{Z} &\gate[style={fill=yellow!20}]{T}&\qw&\qw&\gate[style={fill=yellow!20}]{T}&\gate[style={fill=yellow!20}]{Z} &\gate[style={fill=yellow!20}]{H}&\meter{}   	
		\end{tikzcd}
	\end{adjustbox}	
	\caption{\small{The arrangement of quantum gates in a general IQP circuit.}}
	\label{fig:general_IQP}
\end{figure}
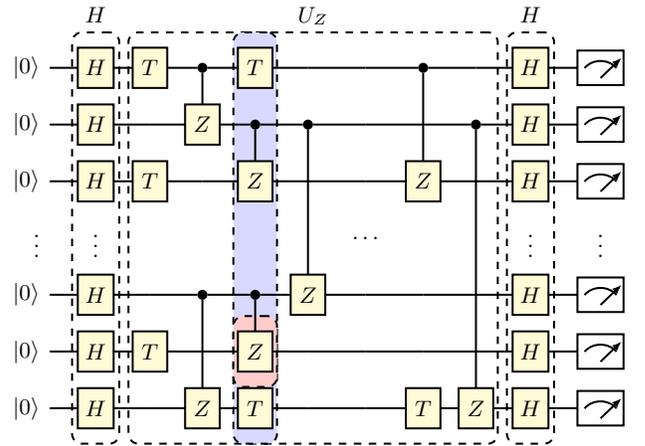 
Before proving that an IQP circuit can be efficiently simulated by MPQC, we first define the arrangement of quantum gates in each block. As shown in Fig.~\ref{fig:MPQC_IQP_block}, from left to right in each block, the seven parameterized single qubit gates are $R_X$, $R_Z$, $R_X$, $R_{\phi}$, $R_{Z}$, $R_Y$ and $R_Z$, followed by $N-1$ CNOT gates, where the controlled qubit of all of them is the first qubit. For simplicity, we will use $(\theta_1,\cdots,\theta_7)$ to represent the composition of the seven parameterized qubit gates. 
		
		{Note that ${H}=R_X(\pi/2)R_Z(\pi/2)R_X(\pi/2)$. Hence, it is not hard to see that the initial and final layers of IQP circuits, where $N$ H gates are separately applied to $N$ qubits, can be simulated by choosing parameters  $(\pi/2,\pi/2,\pi/2,0,0,0,0)$ and $(0,0,0,0,0,0,0)$ in the first and second blocks, respectively.}

	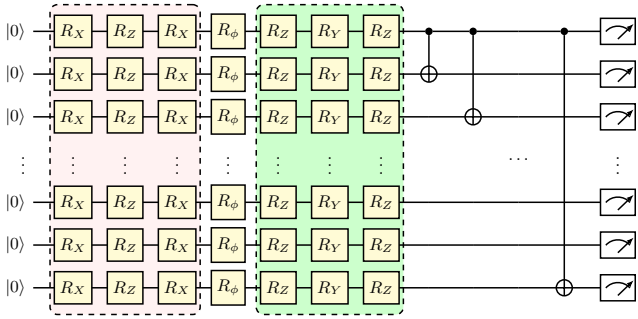
\begin{figure}[!ht]
	\begin{adjustbox}{width=0.48\textwidth}
		\begin{tikzcd}[row sep={0.8cm,between origins}]
			\lstick{$\ket{0}$} &\gategroup[wires=7,steps=3,style={dashed,rounded corners,fill=pink!20,inner xsep=2pt},background]{} \gate[style={fill=yellow!20}]{R_X} &\gate[style={fill=yellow!20}]{R_Z} &\gate[style={fill=yellow!20}]{R_X} &\gate[style={fill=yellow!20}]{R_{\phi}} &\gategroup[wires=7,steps=3,style={dashed,rounded corners,fill=green!20,inner xsep=2pt},background]{}\gate[style={fill=yellow!20}]{R_Z} &\gate[style={fill=yellow!20}]{R_Y} &\gate[style={fill=yellow!20}]{R_Z} &\ctrl{1} &\ctrl{2} &\qw &\ctrl{6}&\meter{}   \\
			\lstick{$\ket{0}$} &\gate[style={fill=yellow!20}]{R_X} &\gate[style={fill=yellow!20}]{R_Z}&\gate[style={fill=yellow!20}]{R_X}&\gate[style={fill=yellow!20}]{R_{\phi}} &\gate[style={fill=yellow!20}]{R_Z} &\gate[style={fill=yellow!20}]{R_Y} &\gate[style={fill=yellow!20}]{R_Z} &\targ{}&\qw&\qw&\qw&\meter{}  \\
			\lstick{$\ket{0}$}	&\gate[style={fill=yellow!20}]{R_X}&\gate[style={fill=yellow!20}]{R_Z}&\gate[style={fill=yellow!20}]{R_X}&\gate[style={fill=yellow!20}]{R_{\phi}} &\gate[style={fill=yellow!20}]{R_Z} &\gate[style={fill=yellow!20}]{R_Y} &\gate[style={fill=yellow!20}]{R_Z}&\qw&\targ{}&\qw&\qw&\meter{}  \\
			\lstick{\vdots}		&\vdots &\vdots &\vdots &\vdots &\vdots	&\vdots	&\vdots	& &	&\ldots &  &\vdots\\
			\lstick{$\ket{0}$}	&\gate[style={fill=yellow!20}]{R_X}&\gate[style={fill=yellow!20}]{R_Z}&\gate[style={fill=yellow!20}]{R_X}&\gate[style={fill=yellow!20}]{R_{\phi}} &\gate[style={fill=yellow!20}]{R_Z} &\gate[style={fill=yellow!20}]{R_Y}&\gate[style={fill=yellow!20}]{R_Z}&\qw&\qw&\qw&\qw&\meter{}  	\\
			\lstick{$\ket{0}$}	&\gate[style={fill=yellow!20}]{R_X}&\gate[style={fill=yellow!20}]{R_Z}&\gate[style={fill=yellow!20}]{R_X}&\gate[style={fill=yellow!20}]{R_{\phi}}	&\gate[style={fill=yellow!20}]{R_Z} &\gate[style={fill=yellow!20}]{R_Y}&\gate[style={fill=yellow!20}]{R_Z}&\qw&\qw&\qw&\qw&\meter{}  \\
			\lstick{$\ket{0}$}	&\gate[style={fill=yellow!20}]{R_X}	&\gate[style={fill=yellow!20}]{R_Z}&\gate[style={fill=yellow!20}]{R_X}&\gate[style={fill=yellow!20}]{R_{\phi}}&\gate[style={fill=yellow!20}]{R_Z} &\gate[style={fill=yellow!20}]{R_Y}&\gate[style={fill=yellow!20}]{R_Z}&\qw&\qw&\qw&\targ{}&\meter{}  
		\end{tikzcd}
	\end{adjustbox}	
	\caption{\small{The arrangement of quantum gates in each block.}}
	\label{fig:MPQC_IQP_block}
\end{figure}

		Next we demonstrate that the internal diagonal matrix $U_Z$ can also be simulated using the predefined block structure. Without loss of generality, we assume that the $i$-th circuit depth in  $U_Z$ contains $M_T$ $T$ gates and $M_{CZ}$ $CZ$ gates, with $M_T+ M_{CZ}\leq N$. For example, the colored region in Figure \ref{fig:general_IQP} indicates that $M_T=2$ and $M_{CZ}=2$. 
		
		Similar to the simulation of H gates, two blocks are sufficient to simulate $M_T$ ($M_T\leq N$) $T$ gates at the same circuit depth. Since ${T}= R_{\phi}(\pi/4)$, then the $T$ gates can be simulated by application of $(0,0,0,\pi/4,0,0,0)$ followed by $(0,0,0,0,0,0,0)$.

		{We next prove how to use predefined blocks to efficiently simulate a $CZ$ gate. Suppose that a $CZ$ gate is  applied to $k$-th qubit, which is controlled by the $j$-th qubit, with $j\leq k$. Since the explicit connection between the two qubits may not exist in the predefined block, we first use $14$ blocks to simulate a SWAP gate that switches the $j$-th controlled qubit to the first qubit. We then use six blocks to simulate the $CZ$ gate that is applied to the $k$-th qubit and controlled by the first qubit. Lastly, $14$ blocks is employed to simulate another SWAP gate to switch the first control qubit back to its original position. For example, as shown in the left panel of Figure \ref{fig:MPQC_IQP_CZ}, the $CZ$ gate as indicated by the {blue} box can be represented by an equivalent circuit {controlled by the first qubit.}} 
		
		\begin{figure}[!ht]
	\begin{adjustbox}{width=0.23\textwidth}
		\begin{tikzcd}[row sep={0.8cm,between origins}]
			&\gategroup[wires=7,steps=1,style={dashed,rounded corners,fill=blue!10,inner xsep=2pt},background]{}\qw&\qw 	\\
			&\qw&\qw 	\\
			&\qw &\qw 	\\
			&\vdots  &\\
			&\ctrl{1} &\qw \\
			&\gategroup[wires=1,steps=1,style={dashed,rounded corners,fill=pink!10,inner xsep=2pt},background]{}\gate[style={fill=yellow!20}]{Z} &\qw \\
			&\qw 	&\qw 	
		\end{tikzcd}\ \ =
		\begin{tikzcd}[row sep={0.8cm,between origins}]
			&\gategroup[wires=7,steps=4,style={dashed,rounded corners,fill=blue!10,inner xsep=2pt},background]{}\qw &\swap{4} &\ctrl{5} &\swap{4} &\qw\\
			&\qw &\qw &\qw	&\qw &\qw	  \\
			&\qw &\qw &\qw	 &\qw&\qw	 \\
			&\vdots & & & &  \\
			&\qw &\targX{} &\qw &\targX{} &\qw\\
			&\qw &\qw	 &\gate[style={fill=yellow!20}]{Z} &\qw&\qw\\
			&\qw &\qw &\qw&\qw&\qw
		\end{tikzcd}
	\end{adjustbox}	
	\begin{adjustbox}{width=0.23\textwidth}
		\begin{tikzcd}[row sep={0.8cm,between origins}]
			&\swap{1}&\qw \\
			&\targX{} &\qw 	
		\end{tikzcd}\ \ =
		\begin{tikzcd}[row sep={0.88cm,between origins}]
			&\ctrl{1}&\targ{} &\ctrl{1}\\
			&\targ{} &\ctrl{-1}&\targ{}
		\end{tikzcd}
	\end{adjustbox}	
	\captionsetup{justification =raggedright}
	\caption{The left panel illustrates an equivalent circuit described by SWAP operation. The right panel shows the implementation of SWAP by two CNOT gates and one reversed CNOT gate.}
	\label{fig:MPQC_IQP_CZ}
\end{figure}
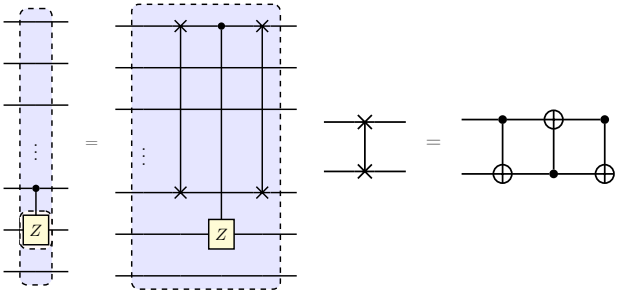

		{The central problem in simulating the SWAP operation is how to simulate a single CNOT gate applied arbitrarily to two qubits, since a SWAP gate is composed of three CNOT gates, as illustrated in the left panel of Figure \ref{fig:MPQC_IQP_CZ}. The first and third CNOT gate of the SWAP operation can be simulated by four blocks. Recall that, in Proposition \ref{prop:1}, a single CNOT gate (namely CX gate) can be decomposed into {{$X=A_1B_1C_1$, where $A_1 = R_Z(0)R_Y(\pi/2)$, $B_1=R_Y(-\pi/2)R_Z(\pi)$, and $C_1=R_Z(\pi)$.}} 
{We set all parameters of the first block as $0$ except the parameters corresponding to the $k$-th qubit, which are set as $(0,0,0,0,0,\pi/2,0)$ to simulate $A_1$. Next, we set all parameters of the second block as $0$ except the parameters corresponding to the $k$-th qubit, which are set as $(0,0,0,0,0,-\pi/2,\pi)$ to simulate $B_1$. Then, we set all parameters of the third block as $0$ except the parameters corresponding to the $k$-th qubit, which are set as $(0,0,0,0,0,0,\pi)$ to simulate $C_1$. Lastly, all parameters of the fourth block are set as $0$. }
}  
		
		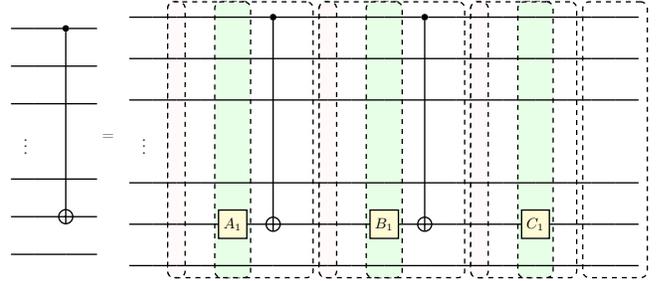
\begin{figure}[!ht]
			\begin{adjustbox}{width=0.48\textwidth}
				\begin{tikzcd}[row sep={0.8cm,between origins}]
					&\qw  &\ctrl{5} &\qw\\
					&\qw   &\qw		&\qw\\
					&\qw	&\qw	&\qw\\
					&\lstick{\vdots}& & \\
					&\qw	&\qw	&\qw\\
					&\qw	&\targ{}	&\qw\\
					&\qw	&\qw	&\qw
				\end{tikzcd}=
				\begin{tikzcd}[row sep={0.88cm,between origins}]
					&\qw&\gategroup[wires=7,steps=5,style={dashed,rounded corners,inner xsep=2pt},background]{}\gategroup[wires=7,steps=1,style={dashed,rounded corners,fill=pink!10,inner xsep=2pt},background]{}\qw &\qw &\gategroup[wires=7,steps=1,style={dashed,rounded corners,fill=green!10,inner xsep=2pt},background]{}\qw  &\ctrl{5} &\qw&\gategroup[wires=7,steps=5,style={dashed,rounded corners,inner xsep=2pt},background]{}\gategroup[wires=7,steps=1,style={dashed,rounded corners,fill=pink!10,inner xsep=2pt},background]{}\qw&\qw &\gategroup[wires=7,steps=1,style={dashed,rounded corners,fill=green!10,inner xsep=2pt},background]{}\qw&\ctrl{5} &\qw &\gategroup[wires=7,steps=4,style={dashed,rounded corners,inner xsep=2pt},background]{}\gategroup[wires=7,steps=1,style={dashed,rounded corners,fill=pink!10,inner xsep=2pt},background]{}\qw&\qw&\gategroup[wires=7,steps=1,style={dashed,rounded corners,fill=green!10,inner xsep=2pt},background]{}\qw &\qw &\gategroup[wires=7,steps=3,style={dashed,rounded corners,inner xsep=2pt},background]{}\qw &\qw &\qw\\
					&\qw&\qw &\qw &\qw &\qw  &\qw&\qw&\qw&\qw&\qw&\qw&\qw&\qw&\qw &\qw &\qw&\qw &\qw\\
					&\qw&\qw&\qw &\qw	 &\qw &\qw&\qw&\qw&\qw&\qw&\qw&\qw&\qw&\qw&\qw &\qw&\qw &\qw\\
					&\lstick{\vdots}& & &  & & & & & & & & &  & & & & \\ 
					&\qw&\qw&\qw &\qw &\qw	&\qw &\qw &\qw &\qw &\qw &\qw&\qw&\qw&\qw&\qw&\qw&\qw &\qw\\
					&\qw&\qw&\qw &\gate[style={fill=yellow!20}]{A_1}&\targ{}&\qw &\qw&\qw&\gate[style={fill=yellow!20}]{B_1}&\targ{}&\qw&\qw&\qw&\gate[style={fill=yellow!20}]{C_1}&\qw&\qw &\qw &\qw\\
					&\qw&\qw&\qw&\qw&\qw &\qw &\qw&\qw&\qw&\qw&\qw&\qw&\qw&\qw&\qw&\qw&\qw&\qw
				\end{tikzcd}
			\end{adjustbox}	
			\caption{\small{Simulating a single CNOT gate by using $4$ blocks.}}
			\label{fig:MPQC_IQPCNOT}
		\end{figure}

		{Six blocks are required to simulate the second reversed CNOT gate (R-CNOT gate) in the SWAP operation. Since $\text{R-CNOT}= (H\otimes H) \text{CNOT} (H\otimes H)$, we use four blocks to simulate the $(H\otimes H) \text{CNOT}$ and then use extra two blocks to simulate the last two Hadamard gates. For the first four blocks, the parameters of the first three parameterized gates that are applied to the first and $i$-th qubits are set as $\pi/2$, $\pi/2$, $\pi/2$, with the aim of simulating two H gates. The remaining parameters of the first four blocks follow with the same setting as simulating the CNOT gate as defined above. The last two blocks follow a similar setting as simulating the Hadamard layer, where the first three parameterized gates that are applied to the first and $i$-th qubits simulate two H gates and the remaining parameters are set as zero. To conclude, a SWAP gate can be composed by a total of $14$ blocks. }

		{Finally, because the $CZ$ gate can be reformulated as $CZ=(\mathbb{I}\otimes H)\text{CNOT}(\mathbb{I}\otimes H)$, it can also be simulated by using six blocks.} 
		
		In summary, since H gates, $T$ gates, and $CZ$ gates can be efficiently simulated by using a constant number of blocks, $O(N)$ blocks are sufficient to simulate an IQP circuit with $O(poly(N))$ $T$ and $CZ$ gates.
	\end{proof}

	\section{Bayesian Quantum Circuit}

In Bayesian inference, additional information about a prior probability distribution $p(\bm{\lambda})$ which represents our beliefs about the parameters of the learning algorithm is given, and the posterior probability distribution $p(\bm{\lambda}|\bm{x})$ can be obtained by Bayes' rule 
\begin{equation}\notag
p(\bm{\lambda}|\bm{x}) = {p(\bm{\lambda})p(\bm{x}|\bm{\lambda})}/{\int_{\bm{\lambda}}p(\bm{\lambda})p(\bm{x}|\bm{\lambda})d\bm{\lambda}},
\end{equation} 	
where $p(\bm{x}|\bm{\lambda})$ is known as the likelihood function. It has been shown that the performance of many learning tasks can be dramatically improved if Bayesian models are employed \cite{cheeseman1988bayesian,kingma2014semi,fei2007learning, kingma2013auto,ghahramani2015probabilistic}.  

Considering the significance of the Bayesian approach in classical machine learning, we devise a Bayesian quantum circuit (BQC) that enables PQCs to accomplish quantum machine learning tasks with Bayesian advantages. We remark that our BQC is the first quantum method for a Bayesian generative model based on PQC. The proposed BQC is capable of explicitly and efficiently generating prior, likelihood, and posterior distributions. Furthermore, we demonstrate that BQC has stronger expressive power than MPQCs studied in previous section.

\subsection{Layouts and Optimization of BQC}\label{subsec:IV2}
Before elaborating BQC, we first define the ancillary driven MPQCs (AD-MPQCs). AD-MPQCs can be divided into two parts, of which the first part aims to generate the targeted distribution and the second part aims to conduct post-selection. In contrast to MPQC, in which all blocks are directly applied to the data qubits, some blocks in AD-MPQC are conditionally applied to the data qubits for specific ancillary quantum states. A general layout of AD-MPQCs is illustrated in Figure \ref{fig:AD-MPQCs}, in which the common shared blocks are highlighted in green and $\ket{\bm{\lambda}}$ represents all possible combinations of $M$ ancillary qubits with $\ket{\bm{\lambda}}=\{\ket{0},\ket{1}\}^{\otimes M}$. 
\begin{figure}[!ht]
	\centering
	\begin{adjustbox}{width=0.42\textwidth,}
		\begin{tikzcd}[row sep={0.8cm,between origins}]
			\lstick{$\ket{0}$}  &\gate[style={fill=pink!20},wires=5]{U_{\bm{\lambda}}^1} &\gate[style={fill=green!20},wires=5]{U^2} &\qw &\gate[style={fill=pink!20},wires=5]{U_{\bm{\lambda}}^{L-1}} &\gate[style={fill=green!20},wires=5]{U^L} &\meter{}\\
			\lstick{\vdots}     &  & &\ldots &  & &\vdots\\
			\lstick{$\ket{0}$}  & & &\qw  &  & &\meter{}\\
			\lstick{$\ket{0}$}   &\qw  &\qw  &\qw  &\qw  &\qw &\meter{}\\
			\lstick{$\ket{0}$}  &\qw   &\qw  &\qw  &\qw   &\qw   &\meter{} \\
			\lstick{$\ket{\bm{\lambda}}$}  &\ctrl{}\vqw{-1} &\qw  &\qw &\ctrl{}\vqw{-1} &\qw &\meter{}                     	
		\end{tikzcd}
	\end{adjustbox}
	\caption{\small{A general framework of AD-MPQC. The arrangement of quantum gates in each block is identical.}}
	\label{fig:AD-MPQCs}
\end{figure}
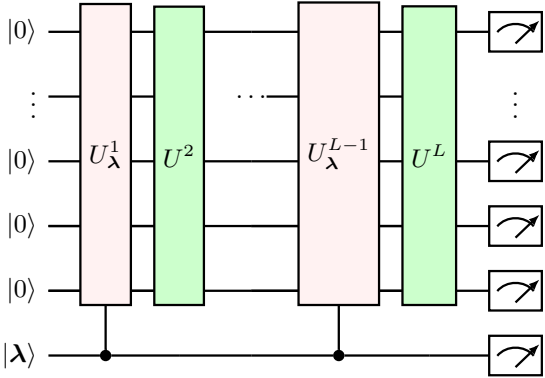

The BQC, in Figure \ref{fig:2}, is a special case of AD-MPQCs in which  the commonly shared blocks (green blocks in Fig.~\ref{fig:AD-MPQCs}) do not exist. In BQC, after applying $K$ blocks $\{U(\bm{\gamma^i})\}_{i=1}^K$ to $M$ ancillary qubits, {the generated state is $\ket{\Psi_A} =\prod_{i=1}^K U(\bm{\gamma^i})\ket{0}^{\otimes M}$.} Measuring the state $\ket{\Psi_A}$ by computational basis, the prior distribution $q(\bm{\lambda}) = |\braket{\bm{\lambda}|\Psi_A}|^2$ is generated. Similarly, after conditionally applying $L$ blocks $\{U(\bm{\theta}_{\lambda_i}^i)\}_{i=1}^L$ to $N$ data qubits iff the ancillary state is $\ket{\lambda_i}$, $\forall\lambda_i\in\bm{\lambda}$, and measuring by computational basis $\ket{\bm{x}}$, the likelihood distribution {$q(\bm{x}|\lambda_i)=|\braket{\bm{x},\lambda_i|\Psi_{\bm{x},\bm{\lambda}}}|^2$} is generated, where  $\ket{\Psi_{\bm{x},\bm{\lambda}}}$ is the quantum state generated by data qubits and ancillary qubits after applying a total of $K+|\bm{\lambda}|L$ blocks.

\begin{figure}[!ht]
	\centering
	\begin{tikzpicture}
	\tikzstyle{every node}=[font=\large]
	\draw [->] [line width=3](-1.1,-2.3) -- (-1.1,-2.9);
	\draw [->] [line width=3](0.5,-2.3) -- (0.5,-2.9);
	\draw [->] [line width=3](1.9,-2.3) -- (1.9,-2.9);
	\draw [->] [line width=3](4.0,-2.3) -- (4.0,-2.9);
	\draw [ultra thick,fill=orange!20] (-2.5,-2.3) rectangle (5.0,-1.6);
	\node at (1.0,-2.0) [thick, font=\fontsize{10}{10}\selectfont, thick]  (first) {Classical Optimizer: Update ($\bm{\theta}/\bm{\gamma}$)};
	\end{tikzpicture}
		\begin{adjustbox}{width=0.48\textwidth}
	\begin{tikzcd}[row sep={0.8cm,between origins}]
		\lstick{$\ket{0}$}   &\qw  &\gate[style={fill=orange!20},wires=4]{U(\bm{\theta}_{\lambda_1}^1)} &\gate[style={fill=orange!20},wires=4]{U(\bm{\theta}_{\lambda_1}^2)} &\qw  &\gate[style={fill=orange!20},wires=4]{U(\bm{\theta}_{\lambda_N}^L)} &\meter{}\\
		\lstick{$\ket{0}$}    &\qw  &\qw &\qw &\qw &\qw &\meter{}\\
		\lstick{\vdots}     &\qw &\qw &\qw &\ldots &\qw \qwbundle &\vdots\\
		\lstick{$\ket{0}$} &\qw	 &\qw &\qw &\qw &\qw &\meter{}\\
		\lstick{$\ket{\bm{\lambda}}$}  &\gate[style={fill=green!20},wires=1]{\{U(\bm{\gamma}^i)\}_{i=1}^K} &\ctrl{}\vqw{-1} &\ctrl{}\vqw{-1} &\qw &\ctrl{}\vqw{-1} &\meter{}  	
	\end{tikzcd}
\end{adjustbox}
	\caption{The general scheme of the proposed BQC.}
	\label{fig:2}
\end{figure}
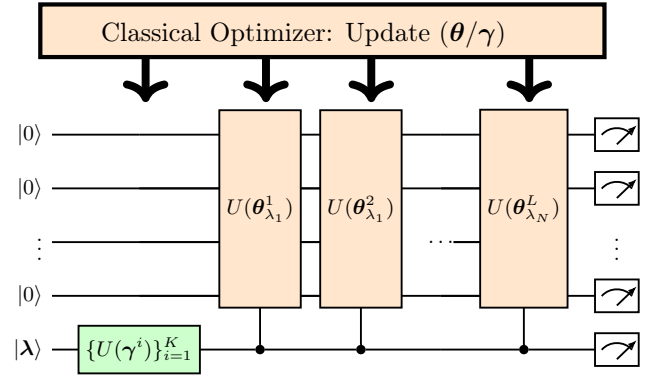

In BQC, the parameterized gates in $U(\bm{\theta}_{\lambda_i}^i)$ are controlled rotational qubits gates, e.g, controlled phase gate $CR_{\phi}(\phi)$, controlled rotation gate along $x$-axis $CR_X(\gamma)$, controlled rotation gate along $y$-axis $CR_Y(\alpha)$, controlled rotation gate along $z$-axis $CR_Z(\theta)$,  which are controlled by the ancillary quantum state $\ket{\bm{\lambda}}$. To reduce the gate complexity, we introduce a flag qubit that is conditionally activated for the specified ancillary state, which enables each parameterized controlled-rotational gate to have only one control qubit. As a result of this extra controlled qubit, the CNOT gates used in MPQCs are replaced by $N$ Toffoli gates. Each Toffoli gate can be efficiently implemented by $10$ single qubit gates and $6$ CNOT gates.  We give an intuitive example of how to apply the block $U(\bm{\theta}_{\lambda_k}^1)$ to the data qubits iff the ancillary state is $\bm{\lambda}_k=\ket{10}$ in Figure \ref{fig:EgBQC}. The green region represents encoding the state $\ket{\Psi_A}$ into ancillary qubits.  The two pink regions represent how to conditionally activate and uncompute the flag qubit for the specific ancillary state $\ket{01}$. The black dotted box illustrates how the block $U(\bm{\theta}_{\lambda_k}^1)$ is conditionally applied to the data register for the specified ancillary state $\ket{\lambda_k} = \ket{01}$.
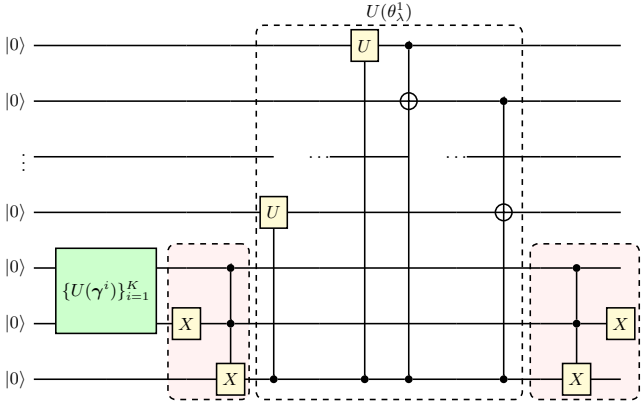
\begin{figure}[!ht]
	\centering
\begin{adjustbox}{width=0.48\textwidth}
	\begin{tikzcd}[row sep={1cm,between origins}]
		\lstick{$\ket{0}$}  &\qw   &\qw &\qw &\gategroup[wires=7,steps=6,style={dashed,rounded corners,inner xsep=2pt},background]{$U(\theta^1_{\lambda}$)}\qw &\qw  &\gate[style={fill=yellow!20}]{U} &\ctrl{}\vqw{1} &\qw &\qw &\qw&\qw&\qw\\
		\lstick{$\ket{0}$}   &\qw  &\qw  &\qw   &\qw &\qw  &\qw &\targ{}  &\qw &\ctrl{}\vqw{2} &\qw &\qw &\qw\\
		\lstick{\vdots}   &\qw  &\qw  &\qw &\qw &\ldots &\qw &\qw &\ldots &\qw &\qw &\qw&\qw\\
		\lstick{$\ket{0}$}  &\qw  &\qw  &\qw &\gate[style={fill=yellow!20}]{U} &\qw  &\qw &\qw &\qw &\targ{} &\qw &\qw&\qw\\
		\lstick{$\ket{0}$}   &\gate[style={fill=green!20},wires=2]{\{U({\bm{\gamma}}^i)\}_{i=1}^K}  &\gategroup[wires=3,steps=2,style={dashed,rounded corners, fill=pink!20, inner xsep=2pt},background]{}\qw &\ctrl{}\vqw{1} &\qw &\qw &\qw  &\qw &\qw &\qw  &\gategroup[wires=3,steps=3,style={dashed,rounded corners, fill=pink!20, inner xsep=2pt},background]{}\qw &\ctrl{1} &\qw\\ 
		\lstick{$\ket{0}$}  &  &\gate[style={fill=yellow!20}]{X}   &\ctrl{}\vqw{1} &\qw &\qw &\qw  &\qw &\qw &\qw &\qw &\ctrl{1} &\gate[style={fill=yellow!20}]{X} \\  
		\lstick{$\ket{0}$}  &\qw &\qw &\gate[style={fill=yellow!20}]{X} &\ctrl{}\vqw{-3} &\qw  &\ctrl{}\vqw{-6} &\ctrl{}\vqw{-5} &\qw &\ctrl{}\vqw{-3} &\qw &\gate[style={fill=yellow!20}]{X} &\qw       	
	\end{tikzcd}
\end{adjustbox}	
	\captionsetup{justification =raggedright}
	\caption{\small{An example of conditionally applying $U(\bm{\theta}^1_{\lambda_k})$ onto data qubits iff $\ket{\lambda_k}=\ket{01}$.}}
	\label{fig:EgBQC}
\end{figure}

In the training process, we employ MMD defined in Eqn. (\ref{eqn:9}) as the loss function. By measuring the data register and the ancillary register, the joint distribution $q(x^i,\bm{\lambda})$ is obtained by  $q(x^i,\bm{\lambda}) = \sum_{\lambda_k\in\bm{\lambda}}\left|\braket{x^i,\lambda_k|\Phi_{\bm{x},\bm{\lambda}}}\right|^2$, where $\ket{\Phi_{\bm{x},\bm{\lambda}}}$ refers to the entanglement quantum states generated by BQC.

\subsection{Expressive Power of BQC and  AD-MPQCs}\label{subsub1}
We first prove that BQC can be formulated by string bond states (SBS) and discuss  the expressive power of BQC and AD-MPQC. By exploiting the connection between BQC and SBS, we prove that if the layout of the quantum gates in each block of AD-MPQCs is allowed to be varied, the AD-MPQC can be efficiently formulated by general tensor networks (GTNs) \cite{glasser2018supervised}.

The central idea in formulating BQC by SBS is to treat all blocks controlled by the same ancillary state as a string operator, as defined in Eqn. (\ref{eqn:SBS}). Given $N$ data qubits and $M$ ancillary qubits, the maximum number of string operators is $|\bm{\lambda}|=2^M$ and the generated quantum state is
\begin{equation}
\ket{\Psi} = \sum_{i=1}^{2^M}\alpha_i\prod_{j=1}^L U({\bm{\theta}^j_{\lambda_i}})\ket{0}^{\otimes N}\ket{\lambda_i},
\end{equation}
where $\alpha_i$ stands for the probability amplitude of state $\ket{\lambda_i}$ with $\sum_i|\alpha_i|^2=1$. Since $\braket{x_i|x_j}=\delta_{ij}$, the generated states corresponding to different ancillary quantum states are independent. Analogous to the string operator $A_{a_i}^{j_i,s}$ defined in Eqn. (\ref{eqn:SBS}) that is conditionally controlled by $s$, this mutually independent property guarantees that block $U({\bm{\theta}^j_{\lambda_i}})$ is conditionally operated with the data register iff the ancillary state is $\ket{x_i}$. 

When there is only one ancillary quantum state $|\bm{\lambda}|=1$, the number of string operators is one and BQC is equivalent to MPQC. This implies that the expressive power of BQC cannot be worse than that of MPQCs. Additionally, since BQC is a special case of AD-MPQCs, the expressive power of AD-MPQC cannot be worse than that of BQC. Therefore, from the perspective of the entanglement entropy,
the expressive power of BQC and AD-MPQCs cannot be worse than that of MPQCs. Since the post-IQP can be efficiently formulated by both AD-MPQCs and BQC, a better expressive power of BQC is obtained compared to MPQCs from the perspective of computational complexity.

The main difference between general tensor networks (GTNs) and regular tensor networks is that GTN allows us to reuse information from a tensor to another part of the network, as also called copy operation \cite{glasser2018supervised}. GTN effectively combines different types of regular tensor networks into one network, which exponentially reduces the number of parameters for describing some functions compared to regular tensor networks. Figure \ref{fig:GTN} (a) gives an example of GTN, which is composed of tree tensor networks (denoted by blue dots) and SBS (denoted by orange and green dots). Two blue arrows indicate the copy operations. Since the essence of the copy operation is independence, i.e., the orange and green dots are independent of each other, AD-MPQCs can
efficiently represent such an independent relation through employing the ancillary register.
As shown in Figure \ref{fig:GTN} (b), if the layout of quantum gates in each block is allowed to be varied, a  quantum circuit corresponding to the GTN illustrated in the Figure \ref{fig:GTN} (a) is constructed.
Although AD-MPQCs can be formulated by GTN, whether there exists some quantum states can be efficiently simulated by AD-MPQCs that are hard for BQC is an open question.
\begin{figure}[!ht]
	\centering
	\raisebox{-.45\height}{				\begin{adjustbox}{width=0.10\textwidth}
			\begin{tikzpicture}
			\Vertex{A}	
			\Vertex[x=1]{B}
			\Vertex[x=2]{C}
			\Vertex[x=3]{D}
			\Vertex[x=0.5, y=1.5]{E}
			\Vertex[x=2.5, y=1.5]{F}
			\Edge(A)(E)
			\Edge(B)(E)
			\Edge(C)(F)
			\Edge(D)(F)
			\Vertex[x=0.5, y=2.8,color=orange!30]{G}
			\Vertex[x=2.5, y=2.8,color=orange!30]{H}
			\Edge[color=black!20](G)(H)
			\Edge[Direct, color=orange!30](E)(G)
			\Edge[Direct, color=orange!30](F)(H)
			\Vertex[x=0.5, y=3.7,color=green!20]{I}
			\Vertex[x=2.5, y=3.7,color=green!20]{J}
			\Edge[color=black!20](I)(J)
			\Edge[Direct, bend=45, color=green!20](E)(I)
			\Edge[Direct, bend=-45, color=green!20](F)(J)
			\end{tikzpicture}
		\end{adjustbox}
	}
	\begin{adjustbox}{width=0.36\textwidth}
		\begin{tikzcd}[row sep={0.8cm,between origins}]
			\lstick{$\ket{0}$}  &\gate[style={fill=yellow!20}]{U} &\ctrl{1} &\gate[style={fill=yellow!20}]{U} &\ctrl{1} &\qw  &\qw &\qw &\qw &\qw &\qw &\meter{}\\
			\lstick{$\ket{0}$}  &\gate[style={fill=yellow!20}]{U} &\targ{} &\gate[style={fill=yellow!20}]{U}  &\targ{}  &\gate[style={fill=yellow!20}]{U} &\ctrl{2}  &\gate[style={fill=yellow!20}]{U} &\ctrl{2} &\ctrl{}\vqw{3} &\ctrl{}\vqw{3} &\meter{}\\
			\lstick{$\ket{0}$}  &\gate[style={fill=yellow!20}]{U} &\ctrl{1} &\gate[style={fill=yellow!20}]{U} &\ctrl{1} &\qw  &\qw  &\qw &\qw &\qw&\qw &\meter{}\\
			\lstick{$\ket{0}$}  &\gate[style={fill=yellow!20}]{U} &\targ{} &\gate[style={fill=yellow!20}]{U}  &\targ{} &\gate[style={fill=yellow!20}]{U} &\targ{} &\gate[style={fill=yellow!20}]{U} &\targ{} &\ctrl{}\vqw{1} &\ctrl{}\vqw{1} &\meter{}\\
			\lstick{$\ket{0}$}  &\qw &\qw  &\qw &\qw &\qw &\qw &\qw  &\qw &\gate[style={fill=pink!20},wires=2]{U_x^1} &\gate[style={fill=pink!20},wires=2]{U_x^2} &\meter{} \\  
			\lstick{$\ket{0}$}  &\qw &\qw  &\qw &\qw &\qw &\qw &\qw   &\qw & & &\meter{}              	
		\end{tikzcd}
	\end{adjustbox}
	\captionsetup{justification =raggedright}
	\caption{The left panel illustrates an example of a general tensor network, composed of tree tensor networks and string bond states. The right panel illustrates the corresponding quantum circuit.}
	\label{fig:GTN}
\end{figure}
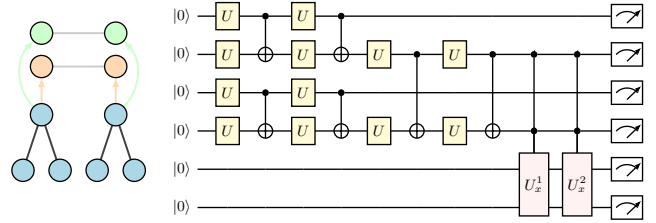

\section{Numerical Experiments}	
\subsection{Generating Bar-and-Stripe Dataset}
To demonstrate the advancements of the proposed BQC, we firstly use BQC to accomplish generative tasks, e.g., generating  $2\times 2$ and $3\times 3$ bars and stripes (BAS) dataset. BAS dataset is composed of vertical bars and horizontal stripes, and some examples of BAS are shown in Figure \ref{fig:3} (a). For $n\times m$ pixels, the number of images that belongs to BAS is $N_{BAS}= 2^n+2^m-2$. The target distribution of such a generative task is denoted as $p(\bm{x})$, where $p(\bm{x}_i) = 1/N_{BAS}$ iff $\bm{x}_i$ is a valid BAS image. The generated probability distribution of BQC $q(\bm{x})=\sum_{\lambda_i\in\bm{\lambda}}q(\bm{x},\lambda_i)$ aims to approximate the targeted distribution $p(\bm{x})$, where the $\bm{x}$ refers to the generated the images, $|\bm{\lambda}|$ refers to the number of valid BAS patterns, and $q(\bm{x},\lambda_i)$ refers to the probability distribution of the generated images given specific $\lambda_i$. 

We compare the generative performance of BQC with two existing MPQCs in the literature, i.e., data driven quantum circuit learning (DDQCL) \cite{benedetti2018generative} and quantum circuit born machine (QCBM) \cite{liu2018differentiable}. Two major differences between DDQCL and QCBM are the layout of CNOT gates in each block and the optimization methods. In DDQCL, the topology of CNOT gates is based on the topology of qantum devices, such as chain, star and all connections. A gradient-free optimization approach is employed, i.e., the swarm optimization algorithm. In QCBM, the topology of CNOT gates is determined by the Chow-Liu tree algorithm, which is inspired by the graphical models to efficiently extract information from training data among different nodes. The unbiased gradient-based optimization approach is employed in the training process. In accordance with the conventions in previous study, in BQC, all BAS patterns are encoded in the qubits, where each data qubit stands for a pixel of the BAS image.

For the task of generating BAS images, the prior is a uniform distribution, since all BAS images are expected to generated with the same probability. Through applying $K$ blocks to the ancillary register with $M$ qubits, the generated quantum state $\ket{\bm{\lambda}}$ is formulated as $\ket{\bm{\lambda}} = \prod_{i=1}^KU(\bm{\gamma^i})\ket{0}^{\otimes M}$, where $q(\bm{\lambda}=\lambda_i)=|\braket{\lambda_i|\bm{\lambda}}|^2=1/{N_{BAS}}$, $|\bm{\lambda}|=N_{BAS}$ and $M=\lceil\log{N_{BAS}}\rceil$.
Since the BAS patterns are encoded into the qubits, the total number of data qubits is $N=n\times m$.  For the specified ancillary state $\bm{\lambda}=\lambda_i$,  $L$ blocks $\{U(\bm{\theta}^i_{\lambda_i})\}_{i=1}^L$ are conditionally applied to the $N$ data qubits, where total $|\bm{\lambda}|L$ blocks are required in BQC. Since there exists a one-to-one mapping that each $\lambda_i$ aims to represent a specific BAS image, we have  $q(\bm{x}=x_i) = q(\bm{x}=x_i,\bm{\lambda}=\lambda_i)$, where $q(\bm{x}=x_i,\bm{\lambda}=\lambda_j)=0$ for $i\neq j$. We remark that it is a special case in generative tasks.

We first train BQC to generate BAS images with $2\times 2$ pixels, where $N_{BAS}=6$ valid images are expected to be generated uniformly after learning. In the experiment, the numbers of data qubits and ancillary qubits are set to be $N=4$ and $M=3$, respectively. Since the prior distribution is known, the parameters of $K=2$ blocks $\{U_j(\bm{\gamma})\}_{j=1}^2$ are fixed, where the generated state is $\prod_{j=1}^2 U_j(\bm{\gamma})\ket{0}^{\otimes M}=1/\sqrt{6}\sum_i\ket{\lambda_i}$ with $\lambda_i\in\bm{\lambda}$ and $|\bm{\lambda}|=6$. In the numerical simulation, we use the function provided by QVM to directly generate the prior distribution $p(\bm{\lambda})$.  In the learning process, we set $L=2$ blocks $\{U_j(\bm{\theta}^i_{\lambda_i})\}_{i=1}^2$ for the specified ancillary quantum state, where each block only contains $4$ $CR_Y(\alpha)$ gates (interacting with $4$ data qubits separately) and the number of Toffoli gates is also $4$ that connect two qubits in sequence, as illustrated in Figure  \ref{fig:EgBQC}. Total $48$ trainable parameters are updated in the learning process.

When BQC is applied to generate  $3\times 3$ BAS images, with $N_{BAS}=14$, the numbers of data qubits $N$ and ancillary qubits $M$ are set as $9$ and $4$, respectively. A uniformly ancillary state is first generated by using the function provided by QVM. Analogous to the $2\times 2$ BAS case, we set $L=2$ and each block contains $9$ $CR_Y(\alpha)$ gates  (interacting with $9$ data qubits separately) and $9$ Toffoli gates. Therefore, total $112$ parameters are updated in the learning process.

Since QVM allows us to read the quantum states directly, the distribution of BAS images can be accessed accurately as measuring infinite times. The experimental results are illustrated in Figure \ref{fig:3}. Here we define the accuracy as $N_{BAS}/N$, where $N$ represents the total number of generated images and $N_{BAS}$ represents  the number of generated images that has  BAS patterns. As shown in Table \ref{tab:1},  BQC outperforms state-of-the-art PQCs, where the accuracy to generate BAS $2\times2$ and $3\times 3$ images is $99.96\%$ and $98.65\%$, respectively. 
\begin{table}[ht!]
	\captionsetup{justification =raggedright}
	\caption{\small{Accuracies for generative 2$\times$2 and $3\times 3$ BAS datasets.}} 
	\centering 
	\begin{tabular}{c c c c c} 
		\hline\hline 
		& Model & DDQCL & QCBM & BQC \\ [0.5ex] 
		\hline 
		$2\times 2$& Accuracy (\%) & 83.82 & 98.46 & 99.96 \\ 
		$3\times 3$& Accuracy (\%) & -- & 65.36 & 98.65\\
		\hline 
	\end{tabular}
	\label{tab:1} 
\end{table}
\begin{figure}[ht!]
	\centering
	\includegraphics[width=0.48 \textwidth]{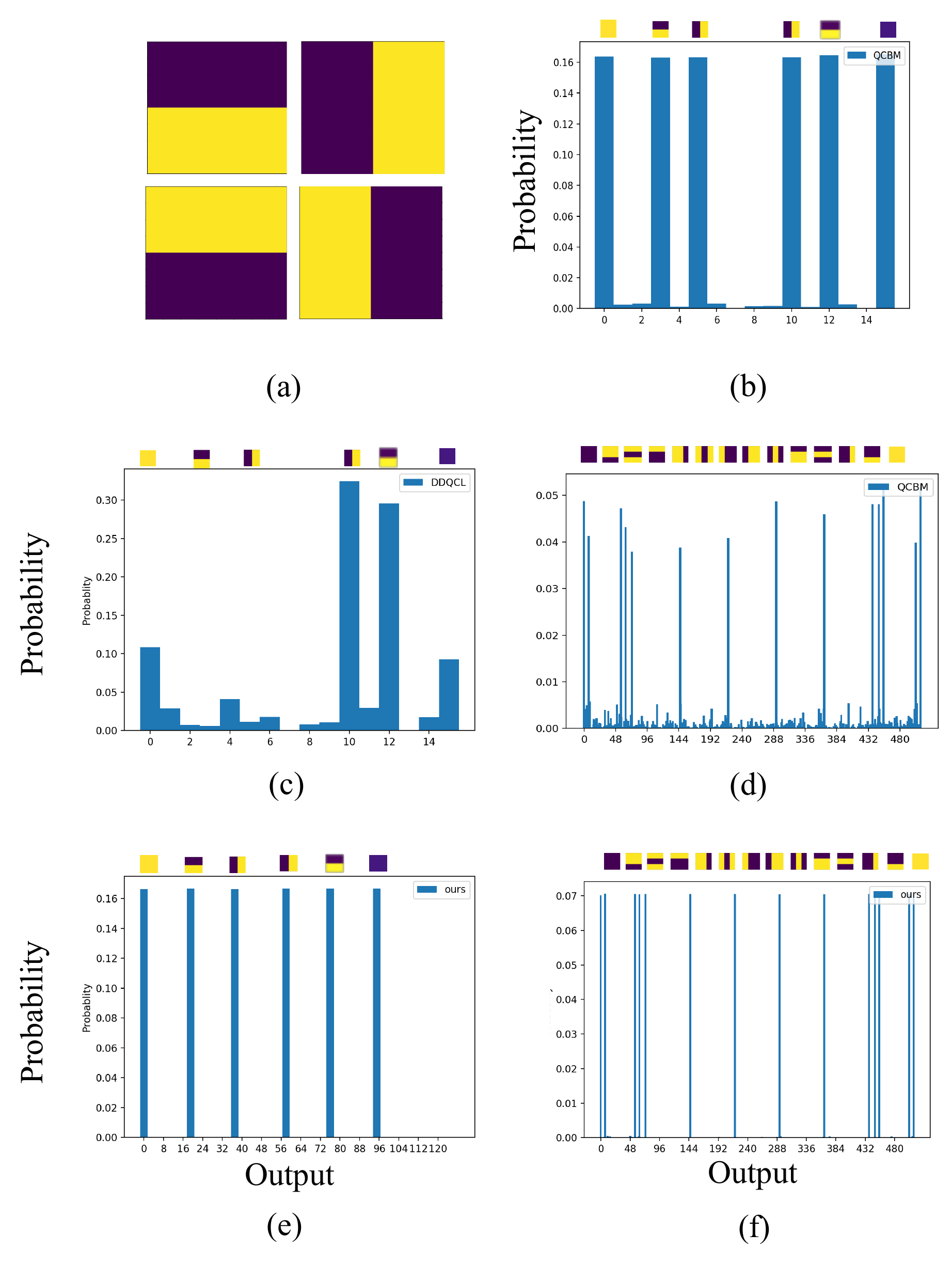}
	\captionsetup{justification =raggedright}
	\caption{\small{The generative results obtained from DDQCL, QCBM, and our model. Since the BAS dataset can be regard as a set of binary images, it can be mapped into different integers, as the x-axis of figures. Figure (b) (c) (e) are the generated result of $2\times 2$ BAS images using QCBM, DDQCL, and BQC respectively. Figure (d) and (f) are the  generated results of $3\times 3$ BAS images using QCBM and BQC, respectively. }}
	\label{fig:3}
\end{figure}

\subsection{Learning Prior Distribution}
How to learn a prior distribution $q(\bm{\lambda})$ efficiently and accurately is one critical topic in machine learning, e.g., to learn the class priors in semi-supervised learning. 
Meanwhile, class priors are also important in learning very sparse data and developing binary classifiers to discriminate positive and unlabeled data \cite{jain2016nonparametric,kemp2004semi}

To confirm the effectiveness of BQC to learn class prior distributions $q(\bm{\lambda})$  from given data, we devise a toy model. Specifically, the training data (referred to the test data with unlabeled class in the above example) are sampled form a joint distribution $p(\bm{x},\bm{\lambda})$, i.e., $
p(\bm{x},\bm{\lambda}) = p(\bm{\lambda})p(\bm{x}|\bm{\lambda})$ with $|\bm{\lambda}|=2$, where the known class conditional densities are $p(\bm{x}|\bm{\lambda}=\lambda_1)\sim\mathcal{N}_1(\mu_1,\sigma_1)$ and $p(\bm{x}|\bm{\lambda}=\lambda_2)\sim\mathcal{N}_2(\mu_2,\sigma_2)$. The $\mathcal{N}_1(\mu_1,\sigma_1)$ and $\mathcal{N}_2(\mu_2,\sigma_2)$ are two Gaussian distributions with means $\mu_1$ and $\mu_2$, variations $\sigma_1$ and $\sigma_2$, respectively. In this toy model, the means and variances of  $\mathcal{N}_1$ and $\mathcal{N}_2$ are set as $u_1=16$, $\mu_2=64$, $\sigma_1=2$, $\sigma_2=4$, respectively. For each class, the known class conditional density distribution is generated by applying $L=7$ blocks $\{U(\bm{\theta}^i)_{\lambda_k}\}_{i=1}^L$to the seven data qubits with $N=7$. Alternatively, $14$ blocks are employed to describe $p(\bm{x}|\bm{\lambda})$ with total $98$ fixed parameters. Since $\bm{x}$ is encoded into qubits as the variable in $\mathcal{N}_1$ and $\mathcal{N}_2$, it is represented as a bit string and should be integers, where the maximum value of $\bm{x}$ is $\bm{x}_{max}=2^N$ and $N$ is the number of qubits.

\begin{figure}[ht!]
	\centering
	\includegraphics[width=3.2in,height=2.4in]{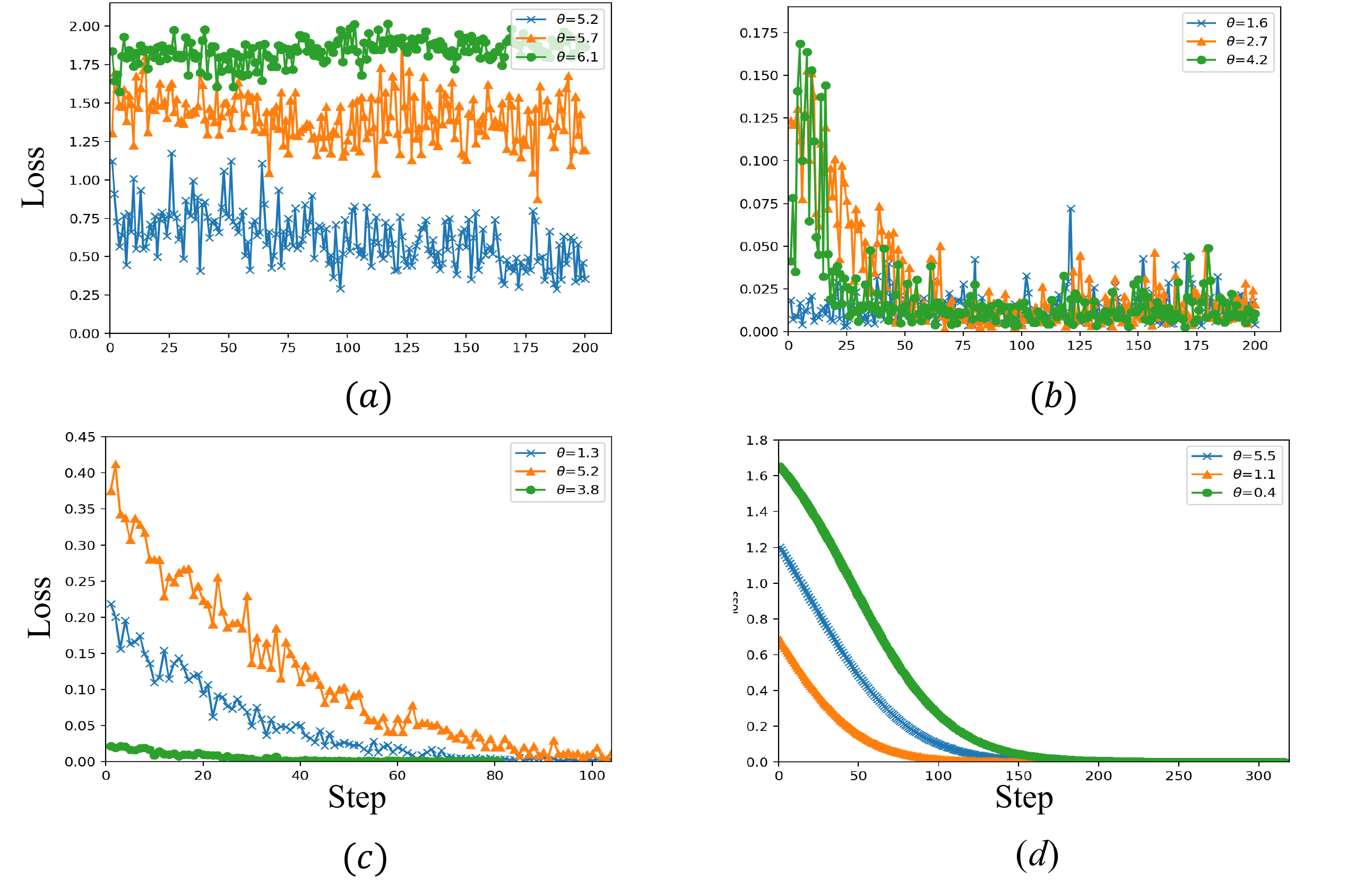}
	\captionsetup{justification =raggedright}
	\caption{\small{The figure illustrates the MMD loss functions for $p(\lambda_1)=0.70$ with different parameter settings. }}
	\label{fig:loss_prior1}
\end{figure} 

\begin{figure}[ht!]
	\centering
	\includegraphics[width=3.2in,height=2.6in]{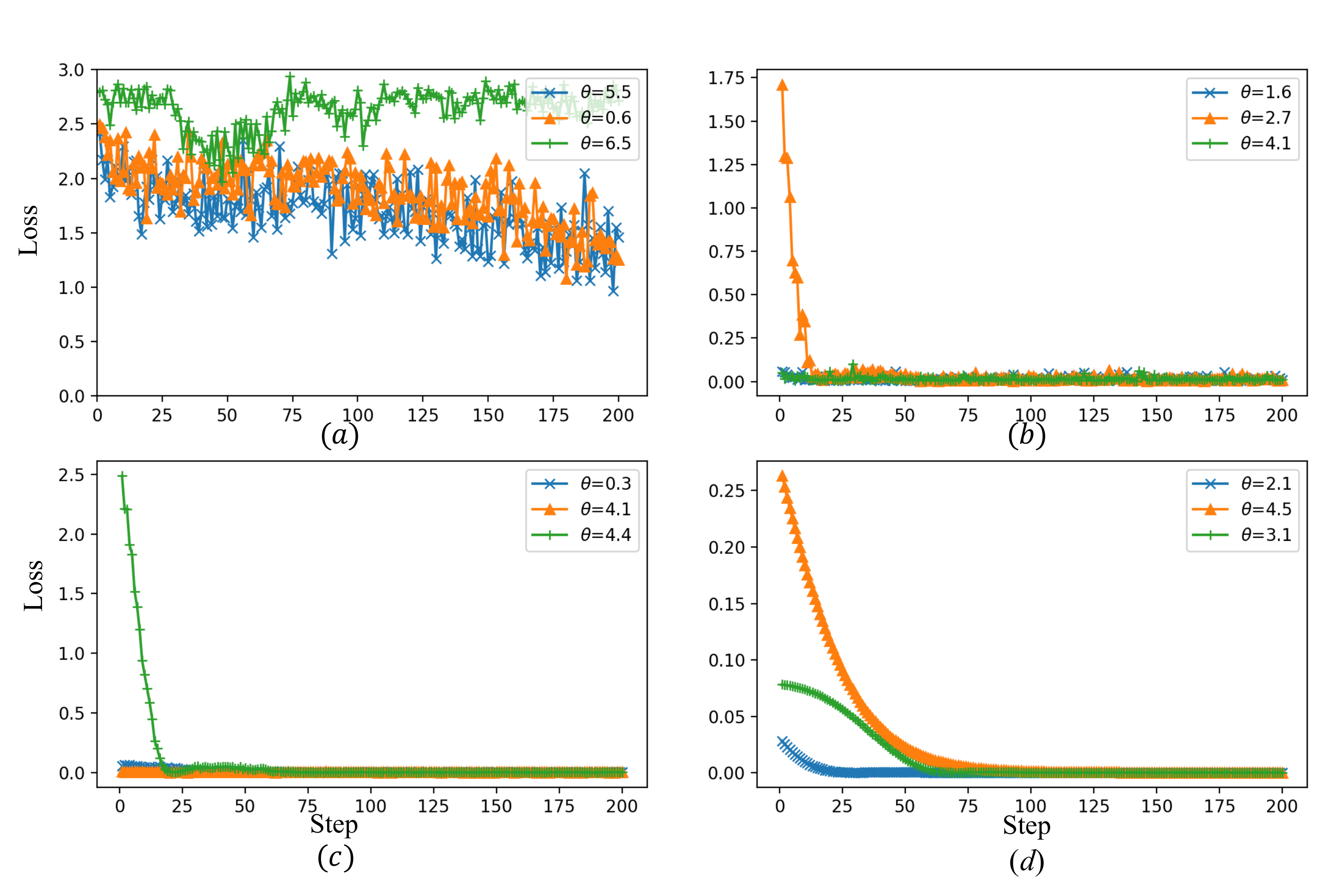}
	\captionsetup{justification =raggedright}
	\caption{\small{The figure illustrates the MMD loss functions for $p(\lambda_1)=0.85$ with different parameter settings.}}
	\label{fig:loss_prior2}
\end{figure} 

In the training process, we estimate two sets of targeted coefficients, i.e., $p(\lambda_1) = 0.7$, $p(\lambda_2) = 0.3$ and $p(\lambda_1) = 0.85$, $p(\lambda_2) = 0.15$, respectively. Due to $|\bm{\lambda}|=2$, we set the number of  ancillary qubit as $M=1$ and employ one block $U(\bm{\gamma}^i)$ to learn the class prior distribution, where the blocks only contains one parameterized $R_y(\alpha)$ gate.  We first use BQC to learn the targeted coefficient $p(\lambda_1) = 0.7$, where the MMD loss functions with three different measurement settings and two gradient descent optimization methods are shown in Figure \ref{fig:loss_prior1}. For each setting, we repeat the experiments three times with varied initialized parameters, which are indicated by different colors.  The training loss by setting the number of measurement as $100$ and employing the general stochastic gradient descent optimization method is shown in Figure \ref{fig:loss_prior1} (a). Then, we employ the unbiased gradient descent method \cite{mitarai2018quantum} and set the number of measurements as $100$, $1000$ and $\infty$, where the training losses are illustrated in \ref{fig:loss_prior1} (b), (c), and (d), respectively. We next use BQC to learn the targeted coefficient $p(\lambda_1) = 0.85$. Following the same parameters setting, the training losses are shown in Figure \ref{fig:loss_prior2}. The two numerical simulation results are listed in Table \ref{tab:2}. The small variance is mainly caused by that the limited parameters $\bm{\theta}$ cannot approximate $\mathcal{N}_1$ and $\mathcal{N}_2$ well.
We remark that different initial parameters have subtle influences to the convergence but the number of measurement determines if the loss can be converged.


\begin{table}[ht!]
\scalebox{0.9}{
	\centering 
	\begin{tabular}{c  c c c c| c c c} 
		\hline\hline 
		Methods  &  Mea & $P(\lambda_1)$ & $P(\lambda_2)$& Variance& $P(\lambda_1)$ & $P(\lambda_2)$& Variance \\ [0.5ex] 
		\hline 
		Target  &--  &0.70 & 0.30& -- & 0.85 &0.15 &--\\
		QVM$^*$    &100 &0.163 &0.867& 4.47E-02&0.555 &0.445 &1.44E-01\\
		QVM    &100 &0.706 &0.294& 1.66E-04&0.868 &0.132 &8.47E-05\\
		QVM   &1000  &0.702 &0.298& 5.74E-06&0.856 & 0.144 &5.94E-06\\
		QVM    &$\infty$ & 0.701 &0.299& 6.12E-10& 0.855& 0.145 &4.67E-09\\
		\hline 
	\end{tabular}
}
\captionsetup{justification =raggedright}
\caption{\small{Learning Prior Distribution with $M=1$, $N= 7$. Here QVM$^*$ stands for employing a general optimization method, while QVM employs the unbiased estimation optimization method.}} 
\label{tab:2}
\end{table}

\section{Conclusion and Discussion}
In this paper, our first contribution is on evaluation of the expressive power of MPQCs, TPQCs and classical neural networks. Characterized by the entanglement entropy, we prove that MPQCs, TPQCs, long range RBM and DBM can efficiently simulate the quantum state satisfying the volume law, which cannot be efficiently simulated by the short range RBM. We next prove that MPQCs can efficiently simulate probability distributions generated by an IQP circuit. These distributions are difficult to simulate efficiently by classical neural networks unless the polynomial hierarchy collapses. We therefore see that MPQCs have stronger expressive power than TPQCs and classical neural networks.

Our second contribution is the proposal of BQC to accomplish Bayesian learning tasks. BQC is a special case of AD-MPQCs that can efficiently simulate the probability distribution generated by post-IQP circuits. It has stronger expressive power over MPQCs without ancillary qubits. In addition, the post-selection operation enables BQC to accomplish machine learning tasks without knowledge about prior distributions. 
We perform two numerical simulations to validate the effectiveness of BQC. The first numerical simulation uses BQC to generate BAS images, in which BQC outperforms state-of-the-art PQCs. The second numerical simulation uses BQC to learn the class prior distribution, which is highly desirable for semi-supervised learning. The simulation results demonstrate that BQC can accurately estimate the prior distributions. These two tasks can be efficiently implemented on near term quantum devices.

Parameterized quantum circuit (PQC) is a hybrid quantum classical learning scheme that has accomplished various learning tasks using a limited number of  quantum gates and a shallow quantum circuit depth.  With the benefit of the strong expressive power and efficient implementation on near-term quantum devices, PQCs have the potential to tackle practical problems with quantum advantages.  One future direction is to explore how to use PQCs to solve practical machine learning problems and to investigate whether the proposed quantum learning model can provide a definitive quantum advantage.  

	\bigskip
\begin{acknowledgments}
	This research is supported by ARC projects FL-170100117 and DP-180103424. MH is supported by an ARC Future Fellowship under Grant FT140100574. The Python codes for our experiments are available upon request. 
\end{acknowledgments}

\bibliography{myref2}

	

\end{document}